%% file: main.tex
\documentclass[11pt]{article}
\usepackage{pgfplots}
\pgfplotsset{compat=1.7}
\usepackage{subcaption}
\usepackage[margin=1in,letterpaper]{geometry}
\usepackage[T1]{fontenc}
\usepackage[utf8]{inputenc}
\usepackage{lmodern}
\usepackage{amsmath}
\usepackage{amssymb}
\usepackage{amsthm}
\usepackage{mathtools}
\usepackage{color,soul}
\usepackage{xcolor}
\usepackage{graphicx}
\usepackage{algorithm}
\usepackage{tcolorbox}
\usepackage{todonotes}
\usepackage[normalem]{ulem}
\usepackage{bbm}
\usepackage{titling}
\usepackage{thmtools}

\usepackage{tkz-graph}
\usepackage{tikz}
\usetikzlibrary{patterns}
\usepackage[noend]{algpseudocode}
\usepackage{hyperref}
\usepackage{cleveref}

\usepackage{comment}

\newcommand{\E}{\mathbb{E}}
\newcommand{\R}{\mathbb{R}}

\input{preamble.tex}

\title{Recovering Communities in Structured Random Graphs}
\date{}
\author{Michael Kapralov \\ EPFL  \and Luca Trevisan  \\ Bocconi University \and Weronika Wrzos-Kaminska \\EPFL}
\begin{document}
\begin{titlingpage}
\maketitle
\abstract{
The problem of recovering planted community structure in random graphs has received a lot of attention in the literature on the stochastic block model, where the input is a random graph in which edges crossing between different communities appear with smaller probability than edges induced by communities. The communities themselves form a collection of vertex-disjoint sparse cuts in the expected graph, and can be recovered, often exactly, from a sample as long as a separation condition on the intra- and inter-community edge probabilities is satisfied. 

In this paper, we ask whether the presence of a large number of overlapping sparsest cuts in the expected graph still allows recovery. For example, the $d$-dimensional hypercube graph admits $d$ distinct (balanced) sparsest cuts, one for every coordinate. Can these cuts be identified given a random sample of the edges of the hypercube where each edge is present independently with some probability $p\in (0, 1)$? We show that this is the case, in a very strong sense: the sparsest balanced cut in a sample of the hypercube at rate $p=C\log d/d$ for a sufficiently large constant $C$ is $1/\text{poly}(d)$-close to a coordinate cut with high probability. This is asymptotically optimal and allows approximate recovery of all $d$ cuts simultaneously. Furthermore, for an appropriate sample of hypercube-like graphs recovery can be made {\em exact}. The proof is essentially a strong hypercube cut sparsification bound that combines a theorem of Friedgut, Kalai and Naor on boolean functions whose Fourier transform concentrates on the first level of the Fourier spectrum with Karger's cut counting argument.
}
\end{titlingpage}
\section{Introduction}
Graph clustering, or community detection, is a fundamental problem in data analysis. The input is a graph $G=(V, E)$, where a subset $C\subseteq V$ of vertices is considered a `community' if it is sparsely connected to the rest of the graph and is reasonably well-connected as an induced subgraph. The task is to recover the communities. 

Graph clustering with a planted solution has received a lot of attention in the literature. In this setting the vertex set $V$ of the graph $G$ is assumed to be partitioned into vertex disjoint clusters $C_1, C_2, \ldots, C_k$ such that the clusters induce well-connected subgraphs and are sparsely connected to the rest of the graph\cite{HLL, SBM11, Mas14, MNS14, MNS15,  BordenaveLM15, ABH16, CPS15, chiplunkar2018testing, GKLMS21}. For example, in the stochastic block model (SBM; \cite{abbe-survey}) edges of $G$ are generated independently, where an edge $\{u, v\}$ is included in the graph with higher probability if $u$ and $v$ belong to the same cluster, and lower probability otherwise. A large body of work on the stochastic block model shows that, if the edge probabilities satisfy a separation condition, the communities $C_1, C_2, \ldots, C_k$ can be recovered from a sample graph with high probability. Determining the exact recovery threshold is a fascinating information theoretic problem for which tight bounds have been obtained over the past two decades \cite{CK01, McSherry, Vu2014ASS, CX16, MNS15, ABH16, AS15}.  Most of the work on SBM has focused on the case of non-overlapping communities, with only a few works allowing for some overlap. At the same time, in the practice of graph clustering one typically does not expects to have very pronounced cluster. Instead, several clusterings of the vertex set may be consistent with the edge set of the graph. Our central question in this paper is:
\begin{center}
{\em Can highly overlapping clusterings be recovered from a sample of the underlying graph? }
\end{center}

Perhaps the most basic example of a graph with a large number of overlapping communities is the hypercube graph on $n=2^d$ vertices, where each of $d$ coordinate cuts is a sparsest cut, and defines a partition of the vertex set into two `communities', namely the two corresponding coordinate halfspaces. This setting is very different from the SBM with two communities, where the expected graph is a union of two cliques on the two clusters and a clique on the entire vertex set, and therefore the two communities are uniquely defined. Formally, the main question we ask in this paper is a structured version of the stochastic block model that allows for many communities with large overlaps:

\begin{center}
{\em Can coordinate cuts be recovered from the edge set of a subsampled hypercube?}
\end{center}

A priori it would seem plausible that cuts of sparsity comparable to the coordinate cuts may emerge in a subsampled hypercube. Intuitively, this could be a mixture of several coordinate cuts in the original cube (a similar effect is seen in rounding the SDP solution to the sparsest cut problem on the hypercube). However, we show that this is not the case:

\begin{thm}\label{thm:cube_main}
Let $Q_d$ be the $d$-dimensional hypercube, and let $Q'_d$ be obtained by including each edge with probability $p \geq C \cdot \log d / d$, where $C$ is a sufficiently large constant. There exists an algorithm with running time $2^{O(n \log n)}$
that, given $Q'_d$, recovers $d$ orthogonal balanced cuts, each with Hamming distance $O(2^d/\mathrm{poly}(d))$ to a coordinate cut, with probability at least $1 - d^{-100}$ over the subsampling.
\end{thm}

We say that a cut $A \subseteq V$ is {\em balanced} if $|A| = |V|/2$. The Hamming distance between two sets $A,B \subseteq V$ is given by $|A \triangle B|$, and we say $A$ and $B$ are \emph{orthogonal cuts} if $|A \triangle B| = |V|/2$. We also remind the reader of the definition of the $d$-dimensional hypercube graph:

\begin{defn}[Hypercube]
	Define the \emph{$d$-dimensional hypercube} to be the graph $Q_{d} = (V,E)$ with vertex set $V = \{0,1\}^d$, and any two vertices are connected by an edge if their Hamming distance is exactly $1$. 
 We let $n \coloneqq |V| = 2^d$ denote the number of vertices. 
\end{defn}

We note that the subsampling rate in Theorem~\ref{thm:cube_main} is such that the expected degree of a vertex is at least $C\log d=C\log\log n$, i.e.\, Theorem~\ref{thm:cube_main} allows for an exponential reduction of the average degree after sampling. The guarantee in Theorem~\ref{thm:cube_main} is tight up to $\poly d$ factors, since the expected number of isolated vertices in $Q'_d$ is at least $2^d/\poly d$. To see this, note that the degree of each vertex in $Q'_d$ is distributed as $\mathrm{Bin}(d,p)$. For $p = C \log d/d$, the probability that a given vertex is isolated is
\[
(1-p)^d = \left(1 - \frac{C \log d}{d}\right)^d \;\geq\; e^{-C \log d}\left(1 - \tfrac{(C\log d)^2}{d}\right) = \frac{1}{\poly d}.
\]
So in expectation, at least $2^d/\mathrm{poly}(d)$ vertices are isolated, and we cannot hope to classify those vertices.

\paragraph{Exact recovery.} Furthermore, we show that, similarly to SBM, exact recovery is possible for a sufficiently high degree sample, specifically, a sample where every vertex has expected degree at least $C\log n, n=2^d,$ for a sufficiently large constant $C>0$. The hypercube itself is not a good model to study this setting, as the degree in the hypercube itself is $d=\log_2 n$. We therefore study the $k$-distance hypercube, defined below:
\begin{defn}[$k$-distance hypercube]
	Define the \emph{$k$-distance hypercube} to be the graph $Q_{d,k} = (V,E)$ with vertex set $V = \{0,1\}^d$, and any two vertices are connected by an edge if their Hamming distance is exactly $k$. 
\end{defn}
When $k$ is odd, the graph $Q_{d,k}$ is connected. When $k$ is even, $Q_{d,k}$ splits into two connected components, corresponding to the vertices of even and odd Hamming weight, respectively: 
\[
Q_d^E \coloneqq \{ x \in \{0,1\}^d : |x| \equiv 0 \pmod{2} \}, 
\qquad 
Q_d^O \coloneqq \{ x \in \{0,1\}^d : |x| \equiv 1 \pmod{2} \},
\]
where \(|x|\) denotes the Hamming weight of \(x\). 

We show that if the sampling rate is such that a given vertex has near-logarithmic expected degree, exact recovery is possible:

\begin{restatable}{thm}{thmthree}\label{thm:general_main}
Let $G=(V,E)$ be a connected component of the $d$-dimensional $k$-distance hypercube $Q_{d,k}$. Let $G' = (V,E')$ be obtained by including each edge with probability $p \geq C \cdot \log d / d^{k-1}$, where $C$ is a sufficiently large constant. There exists an algorithm with running time $2^{O(n \log n)}$ that, given $G'$, exactly recovers the $d$ coordinate cuts, with probability $1 - d^{-100}$ over the subsampling.
\end{restatable}

\paragraph{Related work on community detection on graphs with overlapping communities.} 
A small number of works allow for overlapping communities. \cite{AS15} considers SBM with overlapping communities, and observes that this can be reduced to a standard SBM where each ``community membership profile'' is considered as a separate community. The number of profiles would be too large for our setting, making every vertex in the hypercube have a profile of its own. The work of \cite{OCCAM} considers a variant of the SBM with fractional community memberships and proves asymptotic consistency under strong assumptions, such as the existence of “pure” nodes that only belong to one community. Another line of work considers dense graphs \cite{AGSSG12}, but requires much higher densities than in our setting.

\paragraph{Related work on community detection in geometric random graphs.}
The geometric block model, introduced in \cite{AB17}, generalizes random geometric graphs in the same way that SBMs generalize Erd\H{o}s–Rényi graphs: In this model, vertices are partitioned into communities, randomly embedded in a metric space, and edges are formed as a function of distances and community memberships.
A number of variants and extensions have since been studied \cite{GBM, GBM19,GBM23,GaussianMBM,ABRS20,ABS20,GNW24,SGBM21,SGBM25}. 
These works, however, focus on detecting (non-overlapping) communities and do not provide guarantees for recovering underlying geometric structure. 

Another related direction considers testing whether an observed graph is a realization of an Erd\H{o}s-Rényi random graph or a random geometric graph~\cite{Bubeck2014TestingFH,LMSY22, BG25}. 

\paragraph{Open problem.} Our work leaves open a very exciting open problem of recovering coordinate cuts with the same precision as our results, but in polynomial time. The SoS hierarchy seems to be a promising direction.

\section{Proof of Theorem~\ref{thm:cube_main}}

In this section we prove Theorem~\ref{thm:cube_main}. The proof is algorithmic. We state the (simple) algorithm below, and then proceed to analyze it.  

\paragraph{The algorithm.} Our algorithm finds $d$ orthogonal sparse cuts in the subsampled hypercube by solving the optimization problem  (solved by direct enumeration in time $2^{O(n \log n)}$):
\begin{equation}\label{opt_problem}
 \begin{aligned}
      \min &  \sum_{i = 1}^d |E'(A_i, V \setminus A_i)|   \qquad \text{subject to}&& \\
            & |A_i| = 2^{d-1} && \forall i,   \\
     & |A_i \triangle A_j| = 2^{d-1} &&  \forall i \neq j. \\
    \end{aligned}
\end{equation}
Known results from Fourier analysis show that \emph{before subsampling}, every sparse cut in the hypercube is close (in Hamming distance) to a coordinate cut. Using this, we will prove that \emph{after subsampling}, all cuts that are far from coordinate cuts remain large, and will therefore not be a part of the optimal solution to \eqref{opt_problem}.

 We will need tools from  Fourier analysis of boolean functions, and start by setting up the necessary preliminaries.

\paragraph{Preliminaries.}  
Let $Q_d = (V,E)$ denote the $d$-dimensional hypercube with vertex set $V = \{0,1\}^d$ and edges connecting pairs of vertices that differ in exactly one coordinate. We write $E'$ for the set of edges obtained from $E$ after subsampling. For a subset $A \subseteq V$, we sometimes write $\1_A \in \{0,1\}^V$ for its indicator function, and $\partial(A) = E(A, V \setminus A)$ for its edge boundary. We say that a cut $A \subseteq V$ is \emph{balanced} if $|A|= \frac{|V|}{2}$, and we say that cuts $A,B \subseteq V$ are \emph{orthogonal} if $|A \triangle B|= \frac{|V|}{2}$. 
For $j \in [d]$, $b \in \{0,1\}$, we use the notation  $S_{j,b}$ for the coordinate cut $\{x \in \{0,1\}^d : x_j = b\}$.

For a function $f \colon \{0,1\}^d \to \R$, we define its Fourier transform $\widehat f : 2^{[d]}\rightarrow \R$ by 
$$ \widehat f(S) \coloneqq \mathop{\mathbb{E}}_{x \in \{0,1\}^d}\left[f(x) \chi_S(x)  \right]  = 2^{-d} \langle f, \chi_S\rangle, $$
where $\chi_S$ is the the \emph{Fourier character} $\chi_S(x) = (-1)^{\sum_{i \in S} x_i}$. We call $\widehat f(S)$ the \emph{Fourier coefficient} of $f$ at $S$. 

The Fourier characters form an orthogonal basis for functions on $\{0,1\}^d$, which gives the inverse formula
$$ f(x) = \sum_{S \subseteq [d]}\widehat f(S) \chi_S$$
and Parseval’s identity
$$\sum_{S \subseteq [d]} \widehat f(S)^2 = 2^{-d} \sum_{x \in \{0,1\}^d} f(x)^2.$$

Furthermore, for every $S \subseteq [d]$, the Fourier character $\chi_S$ is an eigenvector with eigenvalue $2|S|$ of the unnormalized Laplacian matrix $\mathcal L = dI - A$ of the hypercube. 

Finally, we need the fact that the singleton cuts are the minimum cuts in a hypercube. 
\begin{lem}[Folklore, see e.g., Example 4.1.3 in \cite{west2001introduction}]\label{lemma:min_cut}
The min cut of the $d$-dimensional hypercube $Q_d = (V,E)$ has size $d$, that is
  $$\min_{A\subseteq V: 1 \leq |A|\leq |V|/2} |E(A, V \setminus A)| = d.$$  
\end{lem}

\paragraph{Every sparse cut is close to a coordinate cut.} We begin by showing that every cut that is sparse in the original hypercube is indeed close to a coordinate cut. For this, we use the following standard Fourier–analytic identity, which expresses the size of a cut in terms of the Fourier coefficients of its indicator function (see e.g. Theorem 2.38 in \cite{ODonnell_2014}). We include a proof for completeness. 
\begin{lem}\label{lemma:fourier-formula_k=1}
    Let $Q_d = (V,E)$ be the $d$-dimensional hypercube, and let $A \subseteq V$. Let $f \colon V \rightarrow \{0,1\}$ denote the indicator function on $A$. Then 
    $$ |E(A, V \setminus A)| = 2^{d+1} \sum_{S \subseteq [d]} |S| \cdot \widehat f(S)^2. $$
\end{lem}

\begin{proof}
Let $\mathcal L$ be the unnormalized Laplacian of $Q_d$. The cut size of $A$ is given by
\begin{equation}\label{eq:L_1}
 |E(A, V  \setminus A)| = \sum_{\{ x,y\} \in E}(f(x) - f(y))^2 =  f^\top \mathcal L f. 
\end{equation}
Expanding $f$ in the Fourier basis gives
$ f = \sum_{S \subseteq [d]} \widehat f(S) \chi_S. $ Since the Fourier characters are eigenvectors of $\mathcal L$ with eigenvalues $2|S|$, we obtain
\begin{equation}\label{eq:L_2}
    f^\top \mathcal L f = \left( \sum_{S\subseteq d} \widehat f(S) \chi_S\right)^\top \mathcal L  \left( \sum_{S\subseteq d} \widehat f(S) \chi_S\right) = \sum_{S \subseteq [d]} 2|S| \widehat f(S)^2\|\chi_S\|^2_2 = 2^{d+1} \sum_{S \subseteq [d]} |S| \cdot \widehat f(S)^2.  
\end{equation}
Combining Equations \eqref{eq:L_1} and \eqref{eq:L_2} gives the lemma. 
\end{proof}
From the above lemma, we will show that every sparse cut must place most of its Fourier mass on the first two levels. This is because the contribution of each Fourier coefficient to the cut size is weighted by $|S|$, so the mass on higher levels contributes to the cut size proportionally.

This will allow us to apply the Friedgut–Kalai–Naor (FKN) theorem, which states that any boolean function with nearly all of its Fourier mass on the first two levels, has to be close to the indicator function of a coordinate cut. 

\begin{thm}[FKN theorem, Theorem 1.1 in \cite{fkn}]\label{thm:fkn}
	If $f \colon \{0,1\}^d \rightarrow \{0,1\}$ is a boolean function, $\| f \|_2^2 = p$ and if $\sum_{|S|>1} \widehat{f}(S)^2 \leq \delta$ then either $p < K' \delta $ or $p > 1 -K'\delta$ or $\|f(x_1, x_2, \dots, x_d)-x_i\| \leq K \delta$ for some $i$ or $\|f(x_1, x_2, \dots, x_d)-(1-x_i)\| \leq K \delta$ for some $i$. Here, $K'$ and $K$ are absolute constants. 
\end{thm}
We remark that the conclusion of the FKN theorem is far from obvious. In particular, the assumption that $f$ is a boolean function is essential. For example, consider a convex combination of coordinate cuts $f(x) = \sum_i \lambda_i \chi_i(x)$. This places all of its Fourier mass on the first two levels, but is in general not close to any coordinate cut. 

A corollary of the FKN theorem (Theorem~\ref{thm:fkn}) is that every sparse cut in the hypercube must be close to a coordinate cut. For completeness, we include a proof of this known fact. 
\begin{lem}[Sparse cuts are close to coordinate cuts, Corollary 1.2 in \cite{fkn}]\label{lemma:sparse-close-to-coord}
Suppose $A \subseteq Q_d$ with $|A| \leq 2^{d-1}$. If $|E(A, V \setminus A)|  \leq (1+\epsilon)|A|$, then there exists a coordinate cut $S_{j,b}$ such that  $$|A \triangle S_{j,b}| \leq 2^{d}K \cdot \epsilon.$$ Here $K$ is an absolute constant.  
\end{lem}
\begin{proof} We start by bounding the Fourier mass above the first two levels in order to apply the FKN theorem (Theorem~\ref{thm:fkn}).
Suppose $|E(A, V \setminus A)|  \leq (1+\epsilon)|A|$. By Lemma~\ref{lemma:fourier-formula_k=1}, 
\begin{equation}\label{eq:E_bound}
  (1+\epsilon)|A| \geq 
 |E(A, V \setminus A)| = 2^{d+1} \sum_{S \subseteq [d]} |S| \cdot \widehat f(S)^2. 
 \end{equation}
On the other hand, 
\begin{equation}\label{eq:A_bound}
 |A| = \sum_{x \in V} f(x)^2 = 2^d \sum_{S \subseteq [d]} \widehat f(S)^2
 \end{equation}
 and, since $ \widehat f(\emptyset) = 2^{-d}\sum_{x \in V} f(x) = 2^{-d}|A|$, we have 

\begin{equation}\label{eq:f_emptyset_bound}
     \widehat f(\emptyset)^2 = 2^{-2d}|A|. 
 \end{equation}
Combining Equations \eqref{eq:E_bound}, \eqref{eq:A_bound} and \eqref{eq:f_emptyset_bound} gives 
\begin{align*}
    \epsilon |A| &\geq  2^{d+1} \sum_{S \subseteq [d]} |S| \cdot \widehat f(S)^2-2^d \sum_{S \subseteq [d]} \widehat f(S)^2 && \text{by \eqref{eq:E_bound} and \eqref{eq:A_bound}} \\
    & = 2^{d+1} \sum_{|S| \geq 1}(|S|-1)\widehat f(S)^2 -2^{d+1}\widehat f(\emptyset)^2+ 2^d\sum_{S \subseteq [d]}\widehat  f(S)^2 \\
    & \geq 2^{d+1} \sum_{|S| \geq 2}\widehat f(S)^2 +2|A| \left(\frac{1}{2}-\frac{|A|}{2^d} \right) && \text{by \eqref{eq:A_bound} and \eqref{eq:f_emptyset_bound}}. 
\end{align*}
The second summand is non-negative by the assumption that $|A| \leq 2^{d-1}$, which gives $\sum_{|S| \geq 2} \widehat f(S)^2 \leq \epsilon 2^{-(d+1)}|A| \leq \epsilon/4$. Therefore, by \cref{thm:fkn} we have $|A \triangle S_{j,b}| \leq 2^d K \cdot \epsilon$ or $|A| \leq K' \epsilon$. Finally, to rule out the second possibility, note that the above equation gives $ \epsilon |A| \geq 2|A|(1/2 - |A|/2^d)$, which rearranges to $|A| \geq 2^{d-1}(1-\epsilon)$.  
\end{proof}

\paragraph{Cut counting on the difference from a coordinate cut.} Next, we want to show that a cut that is far (in Hamming distance) from every coordinate cut, cannot become the sparsest cut after subsampling. 
Let $A$ be a cut with $E(A,V\setminus A)=(1+\epsilon)2^{d-1}$ and $\epsilon > 1/\poly d$. Then $|A\triangle S|\le O(\epsilon)2^{d-1}$ for some coordinate cut $S$. We want to show that with a high probability, $|E'(A,V\setminus A)| > E'|(S,V\setminus S)|$, i.e. that the coordinate cut $S$ is still sparser than $A$ after subsampling. To show this, we want to apply the Chernoff bound to show concentration for each cut, and Karger’s cut-counting theorem to union bound over all possible choices for the cut $A$. 
\begin{thm}[Karger's cut counting theorem \cite{karger}]\label{thm:karger}
	Let $\alpha \geq 1$. Then for all graphs $G$, the number of $\alpha$-approximate minimum cuts in $G$  is at most  $2^{\lceil2\alpha\rceil} {n \choose \lceil2\alpha\rceil}$ . 
\end{thm}

We start by noting that a direct cut counting plus Chernoff bound argument does not work. Indeed, a direct Chernoff bound applied to $E(A, V \setminus A)$ and $E(S, V \setminus S)$ would require concentration within a $(1\pm\epsilon)$ factor {\bf for all $\epsilon>1/\poly d$ simultaneously}, which is too strong. Instead, we use the fact that $A$ is close to $S$, and show that the differences $E'(A,V\setminus A)\setminus E'(S,V\setminus S)$ and $E'(S,V\setminus S)\setminus E'(A,V\setminus A)$ concentrate well. In essence, we apply a Karger-style cut counting argument {\bf on the difference between $A\triangle S$}, thereby only requiring the Chernoff bound to handle a constant factor deviation.

Applying a Chernoff bound, using the trivial upper-bound 
\[
|E(A,V\setminus A)\triangle E(S,V\setminus S)| \le d|A\triangle S| \le d\cdot O(\epsilon)2^{d-1},
\] 
we can show that  $E'(A, V \setminus A)$ and $E'(S, V \setminus S)$ concentrate within a $O(1/d)$-factor with probability at least $1- e^{- \Omega(p \epsilon 2^{d-1}/d)}.$

To union bound, we must enumerate over all cuts $A$ with $|E(A,V\setminus A)|=(1+\epsilon)2^{d-1}$. A naive application of Karger’s theorem (using $\mathrm{mincut}(Q_d)=d$, by Lemma~\ref{lemma:min_cut}) shows that there are at most $2^{O(2^{d-1}/d)} {2^d \choose O(2^{d-1}/d)}$ such cuts, which is too weak of a bound.  

Instead, we observe that for a fixed coordinate cut $S$, the set $A$ is uniquely determined by $A\triangle S$, so it suffices to enumerate the possible choices for $A\triangle S$. Applying Karger's cut-counting theorem  with the trivial bound 
\(|\partial(A\triangle S)| \le d|A\triangle S| \le d\cdot O(\epsilon)2^{d-1}\) 
gives that there are at most $ 2^{O(\epsilon) 2^{d-1}}{ 2^d \choose O(\epsilon) 2^{d-1}} \approx 2^{O(\epsilon) \log 1/\epsilon2^d}$ possible choices for the set $A \triangle S$. However, this bound is still too weak. We will therefore derive a stronger bound on \(\partial(A)\triangle \partial(S)\) and \(\partial(A\triangle S)\).   

\begin{lem}\label{lem:small_cut}
Let $A \subseteq V$ be a set with $|A| \leq 2^{d-1}$ and $|\partial(A)| \leq (1+\epsilon)|A|$ and let $S$ be the coordinate cut such that $|A \triangle S| \leq K \cdot \epsilon 2^{d}$ (exists by Lemma~\ref{lemma:sparse-close-to-coord}). Then
    $$ |\partial(A) \triangle \partial(S)| =  |\partial(A \triangle S)| \leq C\cdot \epsilon 2^{d-1}.  $$
    Here $C$ is an absolute constant. 
\end{lem}
\begin{proof}
It is straightforward to verify that for every pair of sets $T_1, T_2$, it holds that $\partial(T_1) \triangle \partial(T_2) =  \partial(T_1 \triangle T_2)$, which gives the equality $|\partial(A) \triangle \partial(S)| =  |\partial(A \triangle S)| $.  We now prove the inequality $|\partial(A \triangle S)| \leq C\cdot \epsilon 2^{d-1}$. 

Let $A^+ \coloneqq A \setminus S$ and $A^- \coloneqq S \setminus A$. Furthermore, write $\overline S  = V \setminus S$. Then $V$ is partitioned into the four sets $A \cap S$, $A^-$, $A^+$ and $\overline S \setminus A$ (see Figure~\ref{fig:SsetminusA}). 
\begin{figure}[H]
    \centering
    \includegraphics[width=0.35\textwidth]{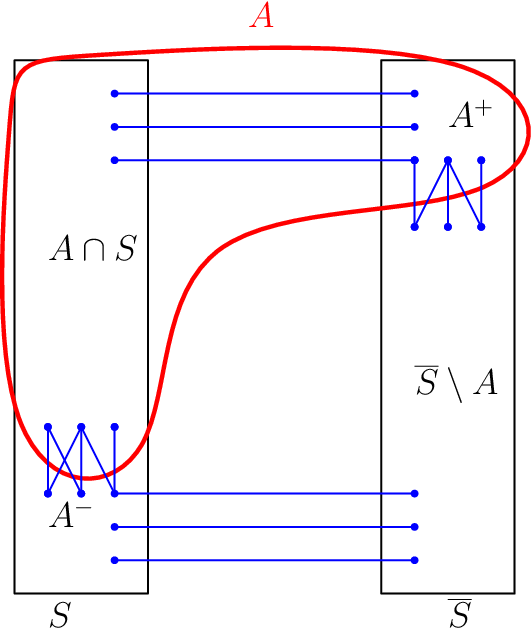}
    \caption{Illustration of the sets $A$ (red), $S$, $A^+$ and $A^-$, and the edges incident on $A^+$ and $A^-$ (blue). }
    \label{fig:SsetminusA}
\end{figure}

The high-level idea is that the edge boundaries of $A^+$ and $A^-$ consist of $E(A^-, \overline S)$ and $E(A^+,S)$, which cross the cut $S$, and $E(A^-, A \cap S)$ and $E(A^+, \overline S \setminus A)$, which cross the cut $A$. Since the former two sets cross the coordinate cut, they have size at most $|A^+| + |A^-| = O(\epsilon)2^{d-1}$. Since the latter two sets contribute to the cut $A$, they cannot be too large, as otherwise the edge-boundary of $A$ would have size significantly larger than $(1+\epsilon)2^{d-1}. $ We now prove this more formally.  
\begin{claim}\label{claim:bdry_main}
\begin{equation*}\label{eq:bdry_main}
    |\partial(A^+)| +  |\partial(A^-)|  \leq |\partial(A)| - |\partial(S)| +2|E(A^+, S)| + 2|E(A^-,\overline S)| . 
\end{equation*}
\end{claim}
\begin{proof}
Since $A$ is partitioned into $A\cap S$ and $A^+$, and $V \setminus A$ is partitioned into $A^-$ and $\overline S \setminus A$, we have 
\begin{equation*}
|\partial(A)| = |E(A \cap S, \overline S \setminus A)| + |E(A^+, A^-)| + |E(A \cap S, A^-)| + |E(A^+, \overline S \setminus A)|.
\end{equation*}
Similarly, since $S$ is partitioned into $A \cap S$ and $A^-$, and $\overline S$ is partitioned into $A^+$ and $\overline S \setminus A$, we have
\begin{equation*}
    |\partial(S)| = |E(A \cap S, \overline S\setminus A)|  + |E(A^-, A^+)|+  |E(A \cap S, A^+)| + |E(A^-, \overline S \setminus A)|.
\end{equation*}
Combining, we get
\begin{align*} 
|\partial(A) | - |\partial(S)|&  =  |E(A^-, A \cap S)| + |E(A^+, \overline S \setminus A)|- |E(A \cap S, A^+)| - |E(A^-, \overline S \setminus A)| \\
& \geq  |E(A^-, A \cap S)| + |E(A^+, \overline S \setminus A)| - |E(A^+, S)| - |E(A^-, \overline S)|. 
\end{align*}
On the other hand, we have 
$$ |\partial(A^-)|+ |\partial(A^+)| =   |E(A^-, S \cap A)| + |E(A^-, \overline S)| + |E(A^+,  S)| + |E(A^+, \overline S \setminus A)|.$$
Combining the above two equations yields the claim. 
\end{proof}

To continue, note that the edges in $E(A^+, S) $ and in $E(A^-, \overline S)$ are crossing the coordinate cut $S$. Since $S$ is a coordinate cut, every vertex can have at most one edge crossing $S$ incident on it. This gives
\begin{equation}\label{eq:bdry_aux}
|E(A^+, S)| + |E(A^-, \overline S)| \leq |A^+| + |A^-| = |A \triangle S| \leq K \cdot \epsilon 2^{d-1},
\end{equation}
where the last inequality follows by the lemma assumption. 
Combining with Claim~\ref{claim:bdry_main}, and recalling from the lemma assumption that $|\partial(A)| \leq (1+ \epsilon)|A| \leq (1+ \epsilon)2^{d-1}$, we obtain 
\begin{align*}
 |\partial(A\triangle S)| & \leq |\partial(A^-)| + |\partial(A^+)| &&  \\
 & \leq |\partial(A)| - |\partial(S)| + 2|E(A^+, S)| + 2|E(A^-, \overline S)|  && \text{by Claim \ref{claim:bdry_main}} \\
 & \leq |\partial(A)| - |\partial(S)| +2K \cdot \epsilon2^{d-1} && \text{by Equation \eqref{eq:bdry_aux}}\\
 & \leq (1+\epsilon)2^{d-1} - 2^{d-1} + 2K \cdot \epsilon2^{d-1}  && \text{by the lemma assumption} \\
 & = C \cdot \epsilon2^{d-1}, && \text{for $C = 2K +1$}
\end{align*}
which completes the proof. 
\end{proof}
With this stronger bound on $|\partial(A\triangle S)|$, we can now bound the number of cuts $A$ of size $\partial(A) \leq (1+\epsilon)2^{d-1}$. 

\begin{lem}\label{lemma:cutcounting}
Let $S$ be a coordinate cut. For every $\epsilon >0$, the number of sets $A \subseteq Q_{k}$ of size $|A| \leq 2^{d-1}$ such that $|\partial(A)| \leq (1+2\epsilon)2^{d-1}$ and $|A \triangle S| \leq K \cdot \epsilon 2^d$
is at most  $ \exp\left(2^d O(\e/d) \log(d/\e)\right)$.
\end{lem}
\begin{proof}
Let $A \subseteq Q_{k}$ be of size $|A| \leq 2^{d-1}$ such that $|\partial(A)| \leq (1+2\epsilon) 2^{d-1}$ and $|A \triangle S| \leq K \cdot \epsilon 2^d$. 
Given $S$, the set $A$ is uniquely determined by the choice of $A\triangle  S$, so we just need to count the number of possible choices for $A\triangle S$. 
Letting $C$ denote the universal constant in Lemma~\ref{lem:small_cut}, we have 
$$|\partial(A \triangle S)|\leq C\epsilon2^{d-1}. $$
\noindent
On the other hand, by Lemma~\ref{lemma:min_cut},  the minimum cut has size $d$, so $\partial(A \triangle S)$ is an $\alpha$-approximate minimum cut with $\alpha =C \epsilon2^{d-1}/d. $ Therefore, by Karger's cut counting theorem (Theorem~\ref{thm:karger}), the number of choices for $|A \triangle S|$ is at most
	$$
	2^{C\e 2^{d-1}/d}{2^d \choose C\e 2^{d-1}/d}\leq 2^{C\e 2^{d-1}/d} \cdot  2^{ H_2(C\e/d)2^d}\leq   \exp\left(2^d O(\e/d) \log(d/\e)\right).$$
    Here $H_2(x)$ denotes the binary entropy function $H_2(x)=-x\log x-(1-x)\log(1-x).$
\end{proof}

To bound the deviation of the cut sizes after subsampling, we use the additive Chernoff bound (see e.g., Theorems 1.10.10 and 1.10.11 in \cite{Doerr}).  
\begin{thm}\label{lemma:chernoff}
	Let $X_1, \dots, X_n$ be independent
	random variables taking values in $[0, 1]$. Let $X =\sum_{i=1}^n X_i$. Let $\lambda \geq 0$. Then 
	$$\Pr[|X - \mathbb{E}[X]| \geq \lambda]\leq 2\exp \left(-\frac{1}{3}\min\left \{ \lambda, \frac{\lambda^2}{\mathbb{E}[X]}\right \}\right).$$
\end{thm}
Applying the Chernoff bound, we show that with high probability,  every cut $A$ remains larger than its closest coordinate cut after subsampling. 
\begin{lem}\label{lemma:cut_concentration}
   Let $\epsilon >0$ and let $A \subseteq V$ be a set of size $A = 2^{d-1}$ such that $(1+\epsilon)|A| \leq E(A, V \setminus A)  \leq (1+ 2\epsilon)|A|.$ Let $S$ be the coordinate cut such that $|A \triangle S| \leq O(\epsilon)2^d$ (exists by Lemma~\ref{lemma:sparse-close-to-coord}). 
   Then 
   $$ \Pr\left[|E'(A,  V\setminus A) | \geq |E'(S, V \setminus S)| + \frac{p\epsilon}{2}\cdot 2^{d-1}. \right]  \geq 1-4e^{-\Omega(\epsilon p2^{d-1})}.$$
\end{lem}
\begin{proof}
    Let $E^+ \coloneqq \partial(A) \setminus \partial(S)$ and let $E^- \coloneqq  \partial(S) \setminus \partial(A) $. Then
\begin{align*}
|E'(A, V \setminus A)| - |E'(S, V \setminus  S)|= \left|( \partial(A) \setminus \partial(S)) \cap E'\right|  - \left|( \partial(S) \setminus \partial(A)) \cap E' \right| = |E^+ \cap E'| - |E^- \cap E'|. 
\end{align*} So we need to bound the probability of the event $|E^+ \cap E'| - |E^- \cap E'| \geq \frac{p \epsilon}{2} 2^{d-1}. $
From the lemma assumption, we have 
\begin{equation}\label{eq:E+E-}
 |E^+| - |E^-| = |E(A, V \setminus A)| - |E(S, V \setminus S)| \geq (1+\epsilon)|A|  - |S| = \epsilon \cdot 2^{d-1}.
 \end{equation}
By Lemma~\ref{lem:small_cut}, 
    $$ \E[|E^+ \cap E'|] = p|E^+| \leq p \cdot O(\epsilon)2^d \qquad \text{and} \qquad \E[|E^- \cap E'|] = p|E^-| \leq p \cdot O(\epsilon)2^d.$$
    Let $\lambda = \frac{p\epsilon}{4}\cdot 2^{d-1}.$ Then $\min\{ \lambda, \lambda^2/p|E^-|\}, \min\{ \lambda, \lambda^2/p|E^+|\} = \Omega\left(\epsilon p 2^{d-1} \right)$, so applying the Chernoff bound (Lemma~\ref{lemma:chernoff}), 
    $$  \Pr\left[\Bigr| |E^-\cap E'|-p|E^-| \Bigr| \geq \frac{p\epsilon}{4}\cdot 2^{d-1}\right] \leq 2e^{-\Omega(\e p 2^{d-1})}$$
and 
    $$  \Pr\left[\Bigr| |E^+\cap E'|-p|E^+| \Bigr| \geq \frac{p\epsilon}{4}\cdot 2^{d-1}\right] \leq 2e^{-\Omega(\e p 2^{d-1})}. $$

By a union bound, with probability at least $1 - 4e^{-\Omega(\e p 2^{d-1})}$, it holds that 
\begin{align*}
    |E^+\cap E'| - |E^- \cap E'| &\geq p|E^+| - p|E^-| - \frac{p \epsilon} {2} \cdot 2^{d-1} \geq p \epsilon 2^{d-1} -  \frac{p \epsilon} {2} \cdot 2^{d-1} = \frac{p \epsilon}{2}\cdot 2^{d-1},
\end{align*}
where the second inequality follows from Equation \eqref{eq:E+E-}.
\end{proof}
We can now show that every sufficiently large cut remains larger than a coordinate cut after subsampling. 
\begin{lem}\label{claim:thm1_claim}
    Let $p \geq \kappa \frac{\log d}{d}$ for a sufficiently large constant $\kappa$,  and let $\epsilon_0 = d^{-100}$. Then with probability at least $1 - d^{-100}/2$, the following holds: For every $\epsilon\geq\epsilon_0$ and every balanced cut $A$ of size $|E(A, V \setminus A)| = (1+\epsilon)2^{d-1}$, it holds that
    $$ |E'(A, V \setminus A)| \geq |E'(S, V \setminus S)| + \frac{p \epsilon}{2}2^{d-1},$$
    where $S$ is the coordinate cut such that $|A \triangle S| \leq K \cdot \epsilon 2^d$ (exists by Lemma~\ref{lemma:sparse-close-to-coord}). 
\end{lem}
\begin{proof} Suppose $\kappa \geq 204  C/D$, where $C$
denotes the hidden constant in the $O$-notation in Lemma~\ref{lemma:cutcounting} and $D$ denotes the hidden constant in the $\Omega$-notation in Lemma~\ref{lemma:cut_concentration}. 
  Let $\epsilon_i = 2^i \epsilon_0$ for $i = 1, \dots,  \log(2^dd/\epsilon_0)$. For every coordinate cut $S=S_{j,b}$ with $j \in [d]$ and $b \in \{0,1\}$, and for every $\epsilon_i$,  let $\mathcal{B}(S,\epsilon_i)$ be the event that there exists a cut $A$ of size $(1+\epsilon_i)2^{d-1} \leq |E(A, V \setminus A)| \leq (1 + 2 \epsilon_i)2^{d-1}$ with $|A \triangle S| \leq K \cdot \epsilon_i2^d$ such that 
  $$ |E'(A, V \setminus A)|< |E'(S, V \setminus S)| + \frac{p \epsilon}{2}2^{d-1}.$$
  We now show that $\Pr[\mathcal{B}(S, \epsilon_i)] \leq d^{-102}/8.$ 
  For every cut $A$ of size $|E(A, V \setminus A)| \leq (1 + 2 \epsilon_i)2^{d-1}$ with $|A \triangle S| \leq K \cdot \epsilon_i 2^d$,  by  Lemma~\ref{lemma:cut_concentration}, we have $$\Pr\left[ |E'(A, V \setminus A)| < |E'(S, V \setminus S)| + \frac{p \epsilon_i}{2}2^{d-1}\right] \leq 4\exp\left(-D\epsilon_i p 2^{d-1}\right).$$
By Lemma~\ref{lemma:cutcounting}, the number of such cuts is at most $\exp\left( 2^d C\frac{\e_i}{d} \log(d/\e_i)\right)$. Therefore, by a union bound, 
\begin{align*}
\Pr[\mathcal{B}(S, \epsilon_i)] & \leq  4\exp\left(-D\epsilon_i p 2^{d-1}\right) \cdot \exp\left( 2^d C\frac{\e_i}{d} \log(d/\e_i)\right)\\
& \leq \exp\left(2^d\epsilon_i \left( C \frac{\log(d/\epsilon_i)}{d}-\frac{D \cdot p}{2} \right)\right)\\
& \leq \exp\left(2^d\epsilon_i \left( 101C \frac{\log d}{d} -  \kappa  \frac{D \log d }{2d}\right)\right) && \text{since $\frac{1}{\epsilon_i} \leq \frac{1}{\epsilon_0} \leq d^{100}$ and $p \geq  \frac{\kappa \log d}{d}$} \\
& \leq \exp\left(- 2^d\epsilon_i C \frac{\log d}{d}  \right) && \text{since $\kappa  \geq 204  C/D $} \\
& \leq \exp\left(-2^d/\poly d\right) && \text{since $\epsilon_i \geq \epsilon_0 =d^{-100}$} \\
&<  \frac{1}{8d^{102}}  && \text{for $d$ sufficiently large.}
\end{align*}
Finally, taking a union bound over the $2d$ possible choices for $S$ and the $\log(2^d/\epsilon_0) \leq 2d$ possible choices for $i$, gives the lemma. 
\end{proof}
Lemma~\ref{claim:thm1_claim} implies that the sparsest cut after sampling is close in Hamming distance to a coordinate cut. However, since the objective value of \eqref{opt_problem} is a sum $\sum_{i}|E'(A_i, V \setminus A_i)|$, we need to exclude the possibility that the optimal solution includes a large cut $|E'(A_i, V \setminus A_i)|$ due to the other cuts being surprisingly sparse. Corollary~\ref{cor:not_too_small} handles this. 
 \begin{cor}\label{cor:not_too_small}
    Conditioned on the success of the event in Lemma~\ref{claim:thm1_claim}, for every balanced cut $A \subseteq V$, it holds that 
    $$|E'(A, V \setminus A)| \geq|E'(S, V \setminus S)|  - C \epsilon_0 2^{d-1},$$
    where $S$ is the coordinate with the smallest hamming distance to $A$ and $C$ is a universal constant.
 \end{cor}
 \begin{proof}
 Let $C$ be the universal constant from \Cref{lem:small_cut}.
     Let $A \subseteq V$ be a cut of size $|E(A, V \setminus A)| = (1+ \epsilon)2^{d-1}, $ and let $S$ be the coordinate cut with the smallest Hamming distance to $A$. By Lemma~\ref{lemma:sparse-close-to-coord}, we have  $|A \triangle S| \leq K \cdot \epsilon 2^{d-1}$. Consider two cases depending on $\epsilon$.
     
    Suppose $\epsilon \geq \epsilon_0$. Then by Lemma~\ref{claim:thm1_claim}, we have 
   $|E'(A, V \setminus A)| \geq|E'(S, V \setminus S)|$, so we are done.
   
   Suppose instead that $\epsilon < \epsilon_0$. Then by Lemma~\ref{lem:small_cut} ,we have   $ |\partial(A)\triangle \partial(S)| \leq C \epsilon_0 2^{d-1}$, which gives
$$ |E'(A, V \setminus A)| \geq |E'(S, V \setminus S)| - |\partial(A)\triangle \partial(S)|  \geq   |E'(S, V \setminus S)| -C \epsilon_0 2^{d-1}. $$
 \end{proof}
 We also need to bound the optimal value of \eqref{opt_problem}. 
\begin{lem}\label{lemma:optvalue} Let $K$ be the universal constant from Lemma~\ref{lemma:sparse-close-to-coord}. 
With probability at least $1-\frac{1}{2d^{100}}$, all coordinate cuts $S_{j,b}$ satisfy
   $$\left| |E'(S_{j,b}, V \setminus S_{j,b})|  - p2^{d-1}\right| \geq \frac{p}{100K d}2^{d-1},$$
  and in particular, the optimal value of \eqref{opt_problem} is at most $\left( d + \frac{1}{100K}\right)2^{d-1}p$.
\end{lem}
\begin{proof}
    Fix a coordinate cut $S_{j,b}$. Then $\E[|E'(S_{j,b}, V \setminus S_{j,b})|] = 2^{d-1}p$, so applying the Chernoff bound (Lemma~\ref{lemma:chernoff}) with $\lambda = \frac{1}{100Kd}2^{d-1}p$ gives
    \begin{align*}
        \Pr\left [\left| |E'(S_{j,b}, V \setminus S_{j,b})|  - p2^{d-1}\right|  \geq \frac{p}{100Kd}2^{d-1} \right]
        & \leq \exp\left(-\frac{1}{3}\frac{1}{100^2K^2d^2}p2^{d-1}\right) \\
        &= \exp\left(-2^d/\poly d\right) \\
      &   \leq d^{-101}/2.  
    \end{align*}
    By a union bound over the $d$ coordinate cuts, the above equation holds simultaneously for all $S_{j,b}$ with probability at least $1 - d^{-100}/2.$
    If this holds, then, since $\{S_{j,b} : j \in [d], b = 0\}$ is a feasible solution to \eqref{opt_problem}, the optimal value of \eqref{opt_problem} is at most 
    $$ \sum_{j} |E'(S_{j,0}, V \setminus S_{j,0})| \leq d \cdot \left( 1 + \frac{1}{100Kd}\right)2^{d-1}p = \left(d +\frac{1}{100K}\right)2^{d-1}p.$$
 \end{proof}

 Finally, we put everything together to prove Theorem~\ref{thm:cube_main}. 
\begin{proof}[Proof of Theorem~\ref{thm:cube_main}] 
The algorithm solves the optimization problem
\begin{equation}\tag{\ref{opt_problem}}
 \begin{aligned}
      \min &  \sum_{i = 1}^d |E'(A_i, V \setminus A_i)|   \qquad \text{subject to}&& \\
            & |A_i| = 2^{d-1} && \forall i,   \\
     & |A_i \triangle A_j| = 2^{d-1} &&  \forall i \neq j \\
    \end{aligned}
\end{equation}
and outputs the optimal solution $A_1,\dots,A_d$. 
\paragraph{Running time.} We can solve this program by enumerating all feasible families of cuts, of which there are at most $\left({2^d \choose 2^{d-1}}\right)^d = 2^{O\left(2^d d\right)} = 2^{O(n \log n)} $, and computing the corresponding edge counts, so the running time is $2^{O(n \log n)}$. 

\paragraph{Correctness.} Condition on the success of the events in Lemma~\ref{claim:thm1_claim} and Lemma~\ref{lemma:optvalue}. By a union bound, this occurs with probability at least $1- d^{-100}.$

Let $\{A_i\}_{i \in [d]}$ be the optimal solution to~\eqref{opt_problem}. For every $i \in [d]$, let $S_i$ denote the coordinate cut with the smallest Hamming distance to $A_i$. 

We start by proving that this is a matching, i.e., that the set $\{S_i\}_{i \in [d]}$ consists of $d$ different coordinate cuts. Suppose not. Then $S_i = S_j$ or $S_i = \overline S_j$ for some $i \neq j$. If $S_i = S_j = S$, then by triangle inequality, 
\[
|A_i \triangle S| + |A_j \triangle S| \;\geq\; |A_i \triangle A_j| \;=\; 2^{d-1},
\]
so either $|A_i \triangle S_i | \geq 2^{d-2}$ or $|A_j \triangle S_j | \geq 2^{d-2}.$

If instead $S_i = \overline S_j = S$, then again by triangle inequality, 
$$ |A_i \triangle S| + |A_j \triangle \overline S| = |A_i \triangle S| + 2^d - |A_j \triangle S|  \geq 2^d-  |A_i \triangle A_j| = 2^{d-1},$$
so again either $|A_i \triangle S_i| \geq  2^{d-2}$ or $|A_j \triangle S_j| \geq 2^{d-2}$. 

Let $i$ be the index such that $|A_i \triangle S_i| \geq 2^{d-2}$. Applying Lemma~\ref{lemma:sparse-close-to-coord} with $\epsilon = 1/4$ gives $\sloppy{|E(A_i, V \setminus A_i)| \geq (1 + \frac{1}{4K})2^{d-1}}$, where $K$ is the universal constant from Lemma~\ref{lemma:sparse-close-to-coord}. Therefore, by Lemma~\ref{claim:thm1_claim} and Lemma~\ref{lemma:optvalue},
\begin{align*}
|E'(A_i, V \setminus A_i)| & \geq |E'(S_i, V \setminus S_i)| + \frac{p}{8K}2^{d-1}   && \text{by Lemma~\ref{claim:thm1_claim}} \\
& \geq \left(1 -  \frac{1}{100K}\right)p2^{d-1} + \frac{p}{8K}2^{d-1}  && \text{by Lemma~\ref{lemma:optvalue}} \\
& \geq \left(1 +  \frac{1}{10K}\right)p2^{d-1}. 
\end{align*}
Furthermore, letting $C$ be the universal constant from Corollary~\ref{cor:not_too_small}, for every $j \neq i$, we have
\begin{align*}
    |E'(A_j, V \setminus A_j)| & \geq |E'(S_j, V \setminus S_j)| - C \epsilon_02^{d-1} && \text{by Corollary~\ref{cor:not_too_small}} \\
    & \geq  \left(1 -  \frac{1}{100Kd}\right)p2^{d-1} -C \epsilon_02^{d-1} && \text{by Lemma~\ref{lemma:optvalue}} \\
    &> \left(1 -  \frac{1}{50Kd}\right)p2^{d-1} && \text{since $\epsilon_0= d^{-100} \ll \frac{p}{Kd}$. }
\end{align*}
But then summing over all $j \in [d]$ gives
\begin{align}
    \sum_{j \in [d]}|E'(A_j, V \setminus A_j)| \geq \left(1 +  \frac{1}{10K}\right)p2^{d-1} + (d-1)\left(1 -  \frac{1}{50Kd}\right)p2^{d-1} > \left(d+ \frac{1}{100K}\right)p2^{d-1},
\end{align}
which is a contradiction, since objective value of \eqref{opt_problem} is at most $\left(d+ \frac{1}{100K}\right)p2^{d}$, by Lemma~\ref{lemma:optvalue}. Thus, the set $\{S_i\}_{i \in [d]}$ must contain $d$ distinct coordinate cuts.

So now suppose that we have a matching, i.e. that the set $\{S_i\}_{ i \in [d]}$ contains $d$ distinct coordinate cuts. Then $\{S_i\}_{i \in [d]}$ is a feasible solution to \eqref{opt_problem}. Recall that $K$ is the universal constant from Lemma~\ref{lemma:sparse-close-to-coord} and $C$ is the universal constant from Corollary~\ref{cor:not_too_small}. Let $L \geq 2 K \cdot C$, and suppose  for contradiction that 
$$|A_i \triangle S_i| \geq L \epsilon_0d2^{d-1}/p$$ for some $i \in [d]$. Applying Lemma~\ref{lemma:sparse-close-to-coord} with $\epsilon = \frac{L}{K}\frac{\epsilon_0d}{p}$, gives
$$|E(A_i, V \setminus A_i)| \geq \left( 1 + \frac{L}{K}  \frac{\epsilon_0 d}{p} \right)2^{d-1}\geq \left(1 + \frac{2C \epsilon_0 d}{p}\right)2^{d-1},$$
where the last inequality follows by choice of $L$. 
So by Lemma~\ref{claim:thm1_claim}, 
\begin{equation*}\label{eq:Ai}
| E'(A_i, V \setminus A_i)| \geq |E'(S_i, V \setminus S_i)| + C \epsilon_0 d 2^{d-1}.
 \end{equation*}
But then, by Corollary~\ref{cor:not_too_small}, 
\begin{align*}
    \sum_{j \in  [d]}|E'(A_j, V \setminus A_j)| & \geq  |E'(S_i, V \setminus S_i)|  + C \epsilon_0 d 2^{d-1} + \sum_{j \neq i} |E'(A_j, V \setminus A_j)|  && \\
    & \geq |E'(S_i, V \setminus S_i)|  + C d \epsilon_0 2^{d-1} + \sum_{j \neq i}|E'(S_j, V \setminus S_j)| - (d-1)C \epsilon_0 2^{d-1}  \\
    & > \sum_{j =1}^d|E'(S_j, V \setminus S_j)|, 
\end{align*}
 which contradicts the optimality of $\{A_i\}_{i \in [d]}$, since $\{S_i\}_{i \in [d]}$ is a feasible solution. We conclude that with probability at least $1-d^{-100}$, it holds that $|A_i \triangle S_i| \leq \frac{L \epsilon_0d2^{d-1}}{p} \leq \frac{2^{d-1}}{\poly d}$ for all  $i$.
\end{proof}
\newpage
\section{Proof of Theorem~\ref{thm:general_main}}
In this section, we prove Theorem~\ref{thm:general_main}, restated below for the convenience of the reader. 
\thmthree*

The proof follows the same overall strategy as Theorem~\ref{thm:cube_main}. The algorithm solves to following optimization problem and outputs the optimal solution. 
\begin{equation}\label{opt_problem_general}
 \begin{aligned}
      \min &  \sum_{i = 1}^d |E'(A_i, V \setminus A_i)|   \qquad \text{subject to}&& \\
            & |A_i| = \frac{|V|}{2} && \forall i,   \\
     & |A_i \triangle A_j| = \frac{|V|}{2} &&  \forall i \neq j. \\
    \end{aligned}
\end{equation}
We want to use the FKN theorem (Theorem~\ref{thm:fkn}) to show that every sparse cut is close to a coordinate cut, and then use Karger's cut-counting theorem (Theorem~\ref{thm:karger}) to union bound over all cuts. In the $k$-distance cube, every vertex has degree ${d \choose k}$, and for every coordinate cut $S_{j,b}$ and every vertex $v \in V$, exactly ${d-1 \choose k-1}$ of the edges incident on $v$ cross the cut $S_{j,b}$. These higher degrees allow for better concentration bounds, which is why we can achieve exact recovery. 
It is important to note that when $k$ is even, the $k$-distance cube $Q_{d,k}$ has at least two connected components, corresponding to the vertices with odd Hamming weight, and the vertices with even Hamming weight.  
\begin{defn}[Component of $Q_{d,k}$]\label{def:component}
 Let  $Q_{d,k}^E$,  $Q_{d,k}^O \subseteq Q_{d,k}$ be the subgraphs induced by   
$$  \{ x \in \{0,1\}^d : |x| \equiv 0 \pmod{2} \}, 
\qquad \text{and} \qquad \{ x \in \{0,1\}^d : |x| \equiv 1 \pmod{2} \},$$
    respectively. 
    We say that $Q \subseteq Q_{d,k}$ is an \emph{component} of $Q_{d,k}$ if 
    \begin{itemize}
        \item $Q = Q_{d,k}$  and $k$ is odd, or
        \item $Q \in \{Q_{d,k}^E, Q_{d,k}^O\}$ and $k$ is even. 
    \end{itemize}
    For a cut $S$, we say that $S$ is a \emph{coordinate cut in $Q$} if $S = S_{j,b} \cap Q$ for some coordinate cut $S_{j,b}. $
\end{defn}
Later, in Remark~\ref{rem:sparsest_n_connected}, we will see that the components of $Q_{d,k}$ are exactly the connected components for $d$ sufficiently large.  

We start by analyzing the spectrum of the $k$-distance cube $Q_{d,k}$. Known results for the Hamming association scheme (see e.g. Theorem 5, Chapter 21 in \cite{macwilliams1977theory}) show that the eigenvalues $\{\mu_S\}_{S \subseteq [d]}$ of the adjacency matrix of $Q_{d,k}$ are given by \emph{binary Krawtchouk polynomials} 
$$ \mu_S = \mathcal{K}_k(|S|; d) \coloneqq \sum_{j = 0}^{k}(-1)^j {|S| \choose j}{d-|S| \choose k-j}.$$
Therefore, the eigenvalues $\{\lambda_S\}_{S\subseteq [d]}$ of the Laplacian $\mathcal L$ satisfy 
$$\lambda_S = {d \choose k} - \mu_S =  2 \sum_{\substack{j \in [k] : \\ j \text{ odd}}} \binom{|S|}{j}\binom{d - |S|}{k-j}.$$
We include a direct calculation of the eigenvalues $\lambda_S$ for completeness. 
\begin{lem}[Eigenvalues of $Q_{d,k}$]\label{lemma:e_vectors_general}
   Let $k$ be an integer, and let $\mathcal{L}$ be the unnormalized Laplacian of $Q_{d,k}$. Then the Fourier characters $\chi_S$ form an eigenbasis of $\mathcal{L}$, with corresponding eigenvalues 
    \begin{equation}\label{eq:lambda_def}
        \lambda_S =  2 \sum_{\substack{j \in [k] : \\ j \text{ odd}}} \binom{|S|}{j}\binom{d - |S|}{k-j}.
  \end{equation}
\end{lem}
\begin{proof}
For every vertex $x \in \{0,1\}^d$, the neighborhood of $x$ in $Q_{d,k}$ is $\{x \oplus \1_T : T \in [d]^{(k)}\}$, where $[d]^{(k)}$ denotes the collection of all $k$-element subsets of $[d]$. Therefore, for every vector $v \in \R^{d}$ and every  $x \in \{0,1\}^d$, 
$$ \mathcal L v(x) = \sum_{T \in [d]^{(k)}}\left(v(x) - v(x \oplus \1_T)\right).$$

 We now show that every Fourier character $\chi_S$ is an eigenvector. For every $S \in [d]$ and every $x \in \{0,1\}^d$, 
\begin{align*}
    \mathcal L \chi_S(x) & = \sum_{T \in [d]^{(k)}}\left(\chi_S(x) - \chi_S (x \oplus \1_T)\right)  \\
    & =  \sum_{T \in [d]^{(k)}}\left( (-1)^{\langle x, \1_S\rangle} - (-1)^{\langle x \oplus\1_T, \1_S\rangle}\right)  \\
    & = (-1)^{\langle x, \1_S\rangle} \sum_{T \in [d]^{(k)}}\left( 1 -(-1)^{\langle \1_T, \1_S \rangle }\right) \\
    & = \chi_S(x) \sum_{T \in [d]^{(k)}}\left( 1 -(-1)^{\langle \1_T, \1_S \rangle }\right). 
\end{align*}
Therefore, $\chi_S$ is an eigenvector of $\mathcal L$ with eigenvalue $\sum_{T \in [d]^{(k)}}\left( 1 -(-1)^{\langle \1_T, \1_S \rangle }\right). $ 
To simplify the expression, note that
$$ \left( 1 -(-1)^{\langle \1_T, \1_S \rangle }\right) = \begin{cases}
    0, & \text{if $|T \cap S|$ is even} \\
    2, & \text{if $|T \cap  S|$ is odd,}
\end{cases}$$
 For every $j \in [k]$, the number of sets $T \in [d]^{(k)}$ such that $|S \cap T| =j$ is ${|S| \choose j}{d-|S| \choose k-j}$, so
\[
    \sum_{T \in [d]^{(k)}}\left( 1 -(-1)^{\langle \1_T, \1_S \rangle }\right) = 2 |\{ T \in [d]^{(k)} : |T \cap S| \text{ is odd}\}|
    = 2 \sum_{\substack{j \in [k] : \\ j \text{ odd}}} \binom{|S|}{j}\binom{d - |S|}{k-j} = \lambda_S,
\]
which completes the proof. 
\end{proof}

Using the above lemma, we can write the size of any cut in $Q_{d,k}$ in terms of the Fourier coefficients of its indicator function.

\begin{lem}\label{lemma:fourier-formula_general}
    Let $Q_{d,k}$ be the $d$-dimensional $k$-distance hypercube, and let $Q =(V,E)$ be component of $Q$ (as per Definition~\ref{def:component}). Let $A \subseteq V$, and let  $f \colon V \rightarrow \{0,1\}$ denote the indicator function on $A$. Then 
    $$ |E(A, V \setminus A)| = 2^{d} \sum_{S \subseteq [d]} \lambda_S \widehat f(S)^2.$$
\end{lem}
\begin{proof}
Since there are no edges between $Q_{d,k}^E$ and $Q_{d,k}^O$, we can use the unnormalized Lapacian $\mathcal L$ of the \emph{entire} graph $Q_{d,k}$ to express the cut size of $A$ as  
\begin{equation}\label{eq:L_1_general}
 |E(A, V  \setminus A)| = \sum_{\{ x,y\} \in E}(f(x) - f(y))^2 =  f^\top \mathcal L f. 
\end{equation}
On the other hand, expanding $f$ in the Fourier basis gives
$ f = \sum_{S \subseteq [d]} \widehat f(S) \chi_S. $ By Lemma~\ref{lemma:e_vectors_general}, every Fourier character $\chi_S$ is an eigenvector of $\mathcal L $ with eigenvalue $\lambda_S$, which gives 
\begin{equation}\label{eq:L_2_general}
    f^\top \mathcal L f = \left( \sum_{S\subseteq d} \widehat f(S) \chi_S\right)^\top \mathcal L  \left( \sum_{S\subseteq d} \widehat f(S) \chi_S\right) = \sum_{S \subseteq [d]} \lambda_S \widehat f(S)^2\|\chi_S\|^2_2 = 2^{d} \sum_{S \subseteq [d]} \lambda_S \cdot \widehat f(S)^2.  
\end{equation}
Combining Equations \eqref{eq:L_1_general} and \eqref{eq:L_2_general} gives the lemma. 
\end{proof}

To argue that every sparse cut places most of its Fourier mass on the first two levels, we first need to argue that the eigenvalues $\lambda_S$ with $|S|>1$ are large compared to those with $|S|=1$. 

\begin{lem}\label{lemma:eigengap}
Let $k$ be a positive integer, and let $d = d(k)$ be sufficiently large.  Denote by $\lambda_1$ the eigenvalue corresponding to sets $S \subseteq [d]$ of size $|S|=1$. Then:
\begin{itemize}
    \item If $k$ is odd, then $
        \lambda_S \ge \tfrac{3}{2}\,\lambda_1$ for every $S \subseteq [d]$ with $|S|\ge 2$.
    \item If $k$ is even, then $ \lambda_S \ge \tfrac{3}{2}\,\lambda_1$ for every $S \subseteq [d]$ with $2 \le |S| \le d-2$, and $\lambda_S = \lambda_1$ for 
 every $S \subseteq [d]$ with $|S| = d-1$. 
\end{itemize}
\end{lem}
\begin{proof} Let $S \subseteq [d]$.  Recall from \eqref{eq:lambda_def} that $\lambda_S =   2 \sum_{\substack{j \in [k] : \\ j \text{ odd}}} \binom{|S|}{j}\binom{d - |S|}{k-j}$ 
and $\lambda_1 =  2 \sum_{\substack{j \in [k] : \\ j \text{ odd}}} \binom{1}{j}\binom{d - 1}{k-j} = 2{ d-1 \choose k-1}$. 
If $k$ is even and $|S| = d-1$, then 
$$ \lambda_S =  2 \sum_{\substack{j \in [k] : \\ j \text{ odd}}} \binom{d-1}{j}\binom{1}{k-j}= {d-1 \choose k-1}= \lambda_1, $$ since the $j = d-1$ term is the only non-zero term in the sum. 

We now show that $\lambda_S \geq 3/2 \lambda_1$ in the remaining cases. We have  \begin{align*}
		\frac{3}{2} \lambda_1 = \frac{3}{2}\cdot 2{d-1 \choose k-1} = \frac{2}{(k-1)!}(d-1)(d-2) \cdots (d-k+1) \leq  \frac{2}{(k-1)!}d^{k-1}
	\end{align*}
so it suffices to show that
\begin{equation}\label{eq:toshow}
    \sum_{\substack{j \in [k] : \\ j \text{ odd}}} {|S| \choose j} {d-|S| \choose k-j}  \geq \frac{1}{(k-1)!} d^{k-1}.
\end{equation}
Write $s \coloneqq |S|$, and consider four cases: The case $ s \leq 2^k$, the case $2^k \leq s \leq d/2$, the case $d/2 < s \leq d$ and $k$ is odd, and the case $d/2 < s \leq d-2$ and $k$ is even.
\paragraph{Case 1: $2 \leq s \leq 2^k$.} 
    Then we can view $s$ as a constant which is independent of $d$. By only considering the $j=1$ term in the sum in \eqref{eq:toshow}, we get 
 \begin{align*}
		\sum_{\substack{j \in [k] : \\ j \text{ odd}}}{s \choose j} {d-s \choose k-j} & \geq s {d-s \choose k-1 } \\
		&  = \frac{s}{(k-1)!}(d-s)(d-(s+1))\cdots (d-(s+k-2))\\
		&\geq \frac{2}{(k-1)!} d^{k-1}-\Omega(d^{k-2}) \\
        & \geq \frac{1}{(k-1)!}d^{k-1},
	\end{align*}
        where the third transition uses that $s$ is a constant independent of $d$, and the last transition holds for $d$ sufficiently large. 

\paragraph{Case 2: $2^k \leq s \leq d/2$.} By only considering the $j=1$ term of the sum in \eqref{eq:toshow}, we get 
\begin{align*}
\sum_{\substack{j \in [k] : \\ j \text{ odd}}}{s \choose j} {d-s \choose k-j} & \geq s {d-s \choose k-1 } \\
	&  \geq \frac{2^{k}}{(k-1)!}\left( d-\frac{d}{2}\right)\left(d-\left(\frac{d}{2}+1\right)\right) \cdots \left(d-\left(\frac{d}{2}+k-2\right)\right)  \\
		&=  \frac{2}{(k-1)!}d(d-2)\cdots (d-2k+4) \\
		&= \frac{2}{(k-1)!}d^{k-1}-\Omega(d^{k-2}) \\
        & \geq  \frac{1}{(k-1)!}d^{k-1},
\end{align*}
where the second transition uses that $s \geq d/2$, and the last inequality holds for $d$ is sufficiently large. 

\paragraph{Case 3: $d/2 \leq s \leq d$ and $k$ is odd.} By only considering the $j = k$ term of the sum  in \eqref{eq:toshow}, we get 
	$$
	\sum_{\substack{j \in [k] : \\ j \text{ odd}}} {s \choose j} {d-s \choose k-j}   \geq {s \choose k} 
	\geq {d/2 \choose k} 
	= \frac{1}{2^k (k-1)!}d^k - \Omega(d^{k-1}) \geq \frac{1}{(k-1)!}d^{k-1},$$
        where the last inequality holds for $d$ sufficiently large. 
\paragraph{Case 4: $d/2 \leq s \leq d-2$ and $k$ is even.}  Let $g(s)$ denote the sum $g(s) =  \sum_{\substack{j \in [k] : \\ j \text{ odd}}} {s \choose j} {d-s \choose k-j }$. Note that if $j$ is odd and $k$ is even, then $k-j$ is odd, so we can relabel the sum to obtain
	$$ g(s) = \sum_{\substack{j \in [k] : \\ j \text{ odd}}}{s \choose j} {d-s \choose k-j }= \sum_{\substack{j \in [k] : \\ j \text{ odd}}}{s \choose k-j} {d-s \choose k-(k-j) } = g(d-s)\geq  \frac{3}{2} {d-1 \choose k-1},$$ where the last inequality follows by applying Case 1 and Case 2 to $d-s$.
\end{proof}

Using the spectral gap established in the previous lemma, we will show that every sparse cut places most of its Fourier mass on the first two levels. As a first step, we derive a lower bound on the expansion of a cut in terms of the higher-level Fourier coefficients.
\begin{lem}\label{lemma:fourier_lb}
	Let $k$ be an integer, let $d$ be sufficiently large, and let $Q=(V,E)$ be a component of $Q_{d,k}$  (as per Definition~\ref{def:component}). Let $A\subseteq Q$  with $|A|\leq |V|/2$, and let $f:\{0,1\}^d \rightarrow \{0,1\}$ be the indicator function on $A$.
    \begin{itemize}
        \item  If $k$ is odd, then $$|E(A,V \setminus A)| \geq {d-1 \choose k-1}\left(  |A| + 2^d  \sum_{S \subseteq [d]: |S| \geq 2} \widehat{f}(S)^2 \right).$$ 
        \item If $k$ is even, then 	$$|E(A, V\setminus A)| \geq {d-1 \choose k-1}\left(  |A| + 2^d \sum_{S \subseteq [d]:2\leq |S| \leq d-2} \widehat{f}(S)^2  \right).$$ 
    \end{itemize}	
\end{lem}
\begin{rem}[Coordinate cuts are sparsest cuts]\label{rem:sparsest_n_connected}
Recall that the coordinate cuts have expansion ${d-1 \choose k-1}$.
Lemma~\ref{lemma:fourier_lb} shows that every cut in the component $Q$ has expansion at least this large, and hence the coordinate cuts are the sparsest cuts. This also implies that each of the components of $Q_{d,k}$ (as per Definition~\ref{def:component}) is connected, so they are exactly the connected components of $Q_{d,k}$.
\end{rem}
\begin{proof}
We have 
\begin{equation}\label{eq:A_bound_general}
    |A| = \sum_{x \in V} f(x)^2  = 2^d \sum_{S \subseteq [d]} \widehat f(S)^2
\end{equation}
and, since $ \widehat f(\emptyset) =  2^{-d} \sum_{x \in V} f(x) = 2^{-d}|A|$, 
\begin{equation}\label{eq:emptyset_bound_general}
\widehat f(\emptyset)^2 =2^{-2d}|A|^2. 
\end{equation}
We now consider the two cases $k$ odd and $k$ even separately. 
\paragraph{Case 1: $k$ odd.} 
Denote by $\lambda_1 = 2{d-1 \choose k-1}$ the eigenvalue corresponding to sets $S \subseteq [d]$ of size $|S|=1$. Note from \eqref{eq:lambda_def}, that $\lambda_{\emptyset} = 0$.
Combining \Cref{lemma:fourier-formula_general}, together with Equations \eqref{eq:A_bound_general} and  \eqref{eq:emptyset_bound_general} gives 
\begin{align*}
    |E(A,V \setminus A)|  &= 2^d\sum_{S \subseteq [d]}\lambda_S \widehat f(S)^2 && \text{by \Cref{lemma:fourier-formula_general}}\\
    & =  2^d\sum_{|S| \geq 2}\left( \lambda_S - \lambda_1\right) \widehat f(S)^2  + 2^d \lambda_1 \sum_{S \subseteq [d]} \widehat f(S)^2 - 2^d \lambda_1 \widehat f(\emptyset)^2 \\
    & \geq 2^d \frac{\lambda_1}{2}\sum_{|S| \geq 2} \widehat f(S)^2   + 2^d \lambda_1 \sum_{S \subseteq [d]} \widehat f(S)^2 - 2^d \lambda_1 \widehat f(\emptyset)^2  && \text{by \Cref{lemma:eigengap}} \\
    & = \frac{\lambda_1}{2} \left( 2^d \sum_{|S| \geq 2} \widehat f(S)^2  +|A| + 2|A|\left( \frac{1}{2} -\frac{|A|}{2^d} \right)\right) && \text{by \eqref{eq:A_bound_general} and \eqref{eq:emptyset_bound_general}} \\
    &\geq {d-1 \choose k-1} \left( 2^d \sum_{|S| \geq 2} \widehat f(S)^2  +|A| \right). 
\end{align*}
Here the last inequality uses $1/2 - 2^{-d}|A| \geq 0$, by the lemma assumption $|A| \leq |V|/2$. 

\paragraph{Case 2: $k$ even.}
Denote by $\lambda_1 = 2{d-1 \choose k-1}$ the eigenvalue corresponding to sets $S \subseteq [d]$ of size $|S|=1$. Note from \eqref{eq:lambda_def}, that $\lambda_{\emptyset} = 0$.
It is not hard to verify that if $f$ is supported only on the odd component $Q_{d,k}^O$, then $\widehat f(S) = -\widehat f([d] \setminus S)$ for all  $S \subseteq [d]$, and if $f$ supported only on the even component $Q_{d,k}^E$, then $\widehat f(S) = \widehat f([d] \setminus S)$  for all $S \subseteq [d]$. In particular,
\begin{equation}\label{eq_f_d_bound}
\widehat f([d])^2 = \widehat f(\emptyset)^2 = 2^{-2d}|A|^2.
\end{equation}
Combining \Cref{lemma:fourier-formula_general}, together with Equations \eqref{eq:A_bound_general} and  \eqref{eq_f_d_bound} now gives 

\begin{align*}
	|E(A,V \setminus A)|  &= 2^d\sum_{S \subseteq [d]}\lambda_S \widehat f(S)^2 && \text{by \Cref{lemma:fourier-formula_general}}\\
	& =  2^d\sum_{2 \leq |S| \leq d-2}\left( \lambda_S - \lambda_1\right) \widehat f(S)^2  + 2^d \lambda_1 \sum_{S \subseteq [d]} \widehat f(S)^2 - 2^d \lambda_1 \left(\widehat f(\emptyset)^ 2 + \widehat f([d])^2\right) \\
	& \geq 2^d \frac{\lambda_1}{2}\sum_{|2 \leq |S| \leq d-2} \widehat f(S)^2   + 2^d \lambda_1 \sum_{S \subseteq [d]} \widehat f(S)^2  - 2^d \lambda_1 \left(\widehat f(\emptyset)^ 2 + \widehat f([d])^2\right) && \text{by \Cref{lemma:eigengap}} \\
	& = \frac{\lambda_1}{2} \left( 2^d \sum_{2 \leq |S| \leq d-2} \widehat f(S)^2  +|A| +2 |A|\left( \frac{1}{2} -\frac{|A|}{2^{d-1}} \right)\right) && \text{by \eqref{eq:A_bound_general} and \eqref{eq_f_d_bound}} \\
	&\geq {d-1 \choose k-1} \left( 2^d \sum_{|S| \leq 2} \widehat f(S)^2  +|A| \right).
\end{align*}
Here the last inequality uses that $1/2 -|A|/2^{d-1}\geq 0 $, by the lemma assumption $|A| \leq |V|/2$. 
\end{proof}

As a corollary of Lemma~\ref{lemma:fourier_lb}, sparse cuts place almost all of their Fourier mass on the first two levels (and, in the even-$k$ case, also on the top two levels).
\begin{cor}\label{cor:bottom_two_levels}
    Let $k$ be an integer, let $d$ be sufficiently large and let $Q=(V,E)$ be a component of $Q_{d,k}$ (as per Definition~\ref{def:component}). Let $A$ be a subset  $Q$  with $|A|\leq |V|/2$ and suppose that $|E(A, V \setminus A)| \leq (1+\epsilon){d-1 \choose k-1}|A|$. Let $f$ denote the indicator function of $A$. Then 
    \begin{itemize}
        \item If $k$ is odd, then $ \sum_{|S| \geq 2} \widehat f(S)^2 \leq \frac{\epsilon}{2}.$
        \item If $k$ is even, then $ \sum_{2 \leq |S| \leq d-2} \widehat f(S)^2 \leq \frac{\epsilon}{2}.$
    \end{itemize}
\end{cor}
\begin{proof}
Suppose that $k$ is odd. By Lemma~\ref{lemma:fourier_lb}, 
$$ (1+\epsilon){d-1 \choose k-1}|A| \geq |E(A, V \backslash A)|  \geq {d-1 \choose k-1}\left(  |A| + 2^d  \sum_{S \subseteq [d]: |S| \geq 2} \widehat{f}(S)^2 \right),$$
which gives 
$$ \epsilon 2^{d-1} \geq \epsilon|A| \geq 2^d  \sum_{S \subseteq [d]: |S| \geq 2} \widehat{f}(S)^2. $$
Suppose instead that $k$ is even. By Lemma~\ref{lemma:fourier_lb}, 
$$ (1+\epsilon){d-1 \choose k-1}|A| \geq |E(A, V \backslash A)|  \geq {d-1 \choose k-1}\left(  |A| + 2^d  \sum_{S \subseteq [d]: 2 \leq |S| \leq d-2} \widehat{f}(S)^2 \right),$$
which gives 
$$ \epsilon 2^{d-1} \geq \epsilon|A| \geq 2^d  \sum_{S \subseteq [d]: 2 \leq |S| \leq d-2} \widehat{f}(S)^2. $$
\end{proof}
We now wish to apply the FKN theorem (Theorem~\ref{thm:fkn}) to argue that every sparse cut must be close to a coordinate cut. We can do that in the case when $k$ is odd. However, when $k$ is even, we are in a slightly different setting, as $f$ also puts Fourier mass also on the top two levels. Therefore, we need to extend the FKN theorem to the case of even $k$. 

\begin{restatable}[FKN theorem for even $k$]{lem}{evenfkn}\label{lemma:fkn_even}
	Let $\1_E$ denote the indicator function of the even component $Q_d^E \coloneqq \{ x \in \{0,1\}^d : |x| \equiv 0 \pmod{2} \}$. 
 Suppose that $f \colon \{0,1\}^d \rightarrow \{0,1\}$ is a boolean function supported on $Q_d^E$ such that $\| f \|_2^2 = \frac{1}{4}$ and $\sum_{2 \leq  |S|\leq d-2} \widehat{f}(S)^2 \leq \delta$. Then  there exists an index $i \in [d]$ such that 
 $\| \1_E \cdot (f(x_1, x_2, \dots, x_d)-x_i)\| ^2_2\leq K \delta$ 
 or $\|\1_E \cdot (f(x_1, x_2, \dots, x_d)-(1-x_i))\|^2_2 \leq K \delta$. Here $K$  is an absolute constant. 
\end{restatable}
The proof is similar to the proof of the FKN theorem in \cite{fkn}, and is included in Section~\ref{sec:fkn_pf}. 

We can now argue that sparse cuts are close to coordinate cuts. 

\begin{lem}[Sparse cuts are close to coordinate cuts]\label{lemma:sparse-close-to-coord_general}
Let $k$ be an integer and let $d = d(k)$ be sufficiently large. Let $Q=(V,E)$ be a component of $Q_{d,k}$ (as per Definition~\ref{def:component}). 
Let $A \subseteq Q$ with $|A| =  |V|/2$. If $|E(A, V \setminus A)|  \leq (1+\epsilon){d-1 \choose k-1}|A|$, then there exists a coordinate cut $S$ in $Q$ (as per Definition~\ref{def:component})such that  $$|A \triangle S| \leq K \cdot \epsilon 2^d.$$ Here $K$ is an absolute constant.  
\end{lem}
\begin{proof}
If $k$ is odd, then the lemma follows from Corollary~\ref{cor:bottom_two_levels} and the FKN theorem (Theorem~\ref{thm:fkn}). 

If $k$ is even, then we can without loss of generality assume that $f$ is supported on the even component, and the lemma follows from Corollary~\ref{cor:bottom_two_levels} and Lemma~\ref{lemma:fkn_even}. 
\end{proof}

Next, we want to apply Karger's cut-counting theorem (Theorem~\ref{thm:karger}) and a Chernoff bound (Lemma~\ref{lemma:chernoff}). 
Similarly to the proof of Theorem~\ref{thm:cube_main}, we need to establish a strong bound on the size of $\partial(A) \triangle \partial(S)$ and $\partial(A \triangle S)$.

\begin{lem}\label{lem:small_cut_general}
Let $Q=(V,E)$ be component of $Q_{d,k}$  (as per Definition~\ref{def:component}) and let $A \subseteq Q$ be a set with $|A| =  |V|/2$ and $|\partial(A)| \leq (1+\epsilon){d-1 \choose k-1}|A|$. Let $S$ be a coordinate cut in $Q$ such that $|A \triangle S| \leq K \cdot \epsilon2^d$ (exists by Lemma~\ref{lemma:sparse-close-to-coord_general}). Then
    $$ |\partial(A) \triangle \partial(S)| = |\partial(A \triangle S)| \leq C\cdot \epsilon { d-1 \choose k-1} |A|. $$
    Here $C$ is an absolute constant. 
\end{lem}
\begin{proof}
The proof is almost identical to the proof of Lemma~\ref{lem:small_cut}. 
It is straightforward to verify that for every pair of sets $T_1, T_2$, it holds that $\partial(T_1) \triangle \partial(T_2) =  \partial(T_1 \triangle T_2)$, which gives the first inequality.  We now prove the equality. 
Let $A^+ \coloneqq A \setminus S$ and $A^- \coloneqq S \setminus A$. Furthermore, write $\overline S  = V \setminus S$. Then $V$ is partitioned into the four sets $A \cap S$, $A^-$, $A^+$ and $\overline S \setminus A$. We have 
\begin{claim}\label{claim:bdry_main_genral}
     $ |\partial(A^+)| +  |\partial(A^-)|  \leq |\partial(A)| - |\partial(S)| +2|E(A^+, S)| + 2|E(A^-,\overline S)| . $
\end{claim}
\begin{proof}
    Identical to the proof Claim~\ref{claim:bdry_main}.
\end{proof}
\noindent
To continue, note that the edges in $E(A^+, S) $ and in $E(A^-, \overline S)$ are crossing the coordinate cut~$S$. Since $S$ is a coordinate cut, every vertex can have at most ${d-1 \choose k-1}$  edges crossing $S$ incident on it. Therefore, 
\begin{equation}\label{eq:bdry_aux_general}
|E(A^+, S)| + |E(A^-, \overline S)| \leq {d-1 \choose k-1}\left( |A^+| + |A^-| \right)= {d-1 \choose k-1}|A \triangle S| \leq K \cdot \epsilon {d-1 \choose k-1}2^{d-1},
\end{equation}
where the last inequality follows from the lemma assumption. 
Combining with Claim~\ref{claim:bdry_main_genral}, and recalling the lemma assumption $|\partial(A)| \leq (1+ \epsilon){d-1 \choose k-1}|A| \leq (1+ \epsilon){d-1 \choose k-1}2^{d-1}$, gives
\begin{align*}
 |\partial(A\triangle S)| & \leq |\partial(A^-)| + |\partial(A^+)| &&  \\
 & \leq |\partial(A)| - |\partial(S)| + 2|E(A^+, S)| + 2|E(A^-, \overline S)|  && \text{by \Cref{claim:bdry_main_genral}} \\
 & \leq |\partial(A)| - |\partial(S)| +2K \cdot \epsilon{d-1 \choose k-1}2^{d-1} && \text{by Equation \eqref{eq:bdry_aux_general}}\\
 & \leq (1+\epsilon){d-1 \choose k-1}|A| - {d-1 \choose k-1}|S| + 2K \cdot \epsilon {d-1 \choose k-1}2^{d-1}  && \text{by the lemma assumption} \\
 & \leq C \cdot \epsilon {d-1 \choose k-1} 2^{d-1} && \text{for $C = 2K +2$},
\end{align*}
where the last inequality uses the assumption that $|A| = |S| = |V|/2 \leq 2^{d-1}$. 
\end{proof}
To apply Karger's cut-counting theorem, we also need to establish the size of the minimum cuts. 

\begin{lem}[The singleton cuts are min-cuts]\label{lemma:min_cut_general}
Let $k$ be an integer, and let $d$ be sufficiently large. Let $Q=(V,E)$ be a component of $Q_{d,k}$ (as per Definition~\ref{def:component}). Then $$\min_{A\subseteq V: 1 \leq |A|\leq |V|/2} |E(A, V \setminus A)| = {d \choose k}.$$  
\end{lem}
\begin{proof}
If $|A| = 1$, then $|E(A, V \setminus A)| = {d \choose k}.$ 
For the rest of the proof, we consider sets $A$ of size $2 \le |A| \le |V|/2$ and split into two cases according to $|A|$.
\paragraph{Case 1: $|A| \geq { d \choose k}/{d-1 \choose k-1}$.}
By Lemma~\ref{lemma:fourier_lb}, 
$$ |E(A, V \setminus A)| \geq {d-1 \choose k-1}|A| \geq {d \choose k}.$$
\paragraph{Case 2: $2 \leq |A| <  { d \choose k}/{d-1 \choose k-1}$.}
Each vertex in $Q$ has degree $d \choose k$. Every vertex in $A$ can have at most $|A|-1$ neighbors in $|A|$, so it must have at least ${d \choose k}-|A|+1$ edges to $V \setminus A$. Therefore,
$$ |E(A, V \setminus A)|  \geq |A|\cdot \left({d \choose k}-|A|+1\right) = \left( {d \choose k}-|A| \right) \left(|A|-1\right)+ {d \choose k} > {d \choose k},$$
where the last inequality follows from ${d \choose k} > { d \choose k}/{d-1 \choose k-1}  \geq |A| \geq 2$. 
\end{proof}
We can now apply Karger's cut-counting theorem to count the number of cuts of size $(1+\epsilon){d-1 \choose k-1}2^{d-1}$.
\begin{lem}\label{lemma:cutcounting_general}
Let $Q=(V,E)$ be a  component of $Q_{d,k}$ (as per Definition~\ref{def:component}), and let $S$ be a coordinate cut in $Q$. For every $\epsilon >0$, the number of sets $A \subseteq Q$ of size $|A| = |V|/2$ such that $|E(A, V \setminus A)| \leq (1+2\epsilon){d-1 \choose k-1}|A|$ and $|A \triangle S| \leq K \cdot \epsilon 2^d$
is at most $\exp\left(2^d O(\e/d) \log(d/\e)\right)$.
\end{lem}
\begin{proof}
Let $A \subseteq Q_{k,d}$ be of size $|A| = |V|/2$ such that $|E(A, V \setminus A)| \leq (1+2\epsilon){d-1 \choose k-1} 2^{d-1}$ and $|A \triangle S| \leq K \cdot \epsilon 2^d$. 

Given $S$, the set $A$ is uniquely determined by the choice of $A\triangle  S$, so we just need to count the number of possible choices for $A\triangle S$.  By Lemma~\ref{lem:small_cut_general}, 
$$|\partial(A \triangle S)| \leq C \cdot \epsilon{d-1 \choose k-1}2^{d-1}. $$
By Lemma~\ref{lemma:min_cut_general},  the minimum cut has size ${d \choose k}$, so $\partial(A \triangle S)$ is an $\alpha$-approximate minimum cut with $\alpha = C\epsilon{d -1\choose k-1}2^{d-1}/{d \choose k} = Ck\epsilon2^{d-1}/d$. Therefore, by Karger's cut counting theorem (Theorem~\ref{thm:karger}), the number of choices for $|A \triangle S|$ is at most
	$$
	2^{Ck\e 2^{d-1}/d}{2^d \choose Ck\e/ 2^{d-1}}\leq 2^{Ck\e2^{d-1}/d} \cdot  2^{ H_2(Ck\e/d)2^d}\leq  \exp\left( 2^d O(\e/d) \log(d/\e)\right).
	$$
Here $H_2(x)$ denotes the binary entropy function $H_2(x)=-x\log x-(1-x)\log(1-x).$
\end{proof}

\begin{lem}\label{lemma:cut_concentration_general}
    Let $Q=(V,E)$ be a component of $Q_{d,k}$ (as per Definition~\ref{def:component}), and let $A \subseteq Q$ be a set with $|A| = |V|/2$ and $(1+\epsilon){d-1 \choose k-1}|A| \leq |\partial(A)|  \leq (1+ 2\epsilon){d-1 \choose k-1}|A|.$ Let $S$ be the coordinate cut in $Q$ such that $|A \triangle S| \leq K\cdot  \epsilon 2^d$ (exists by Lemma~\ref{lemma:sparse-close-to-coord_general}). 
   Then 
   $$ \Pr\left[|E'(A,  V\setminus A) | \geq |E'(S, V \setminus S)| + \frac{p\epsilon}{2} {d-1 \choose k-1}\frac{|V|}2 \right]  \geq 1-4e^{-\Omega(\epsilon p d^{k-1}2^{d-1})}.$$
\end{lem}
\begin{proof}
    Let $E^+ \coloneqq \partial(A) \setminus \partial(S) $ and let $E^- \coloneqq  \partial(S) \setminus \partial(A) $. Then
\begin{align*}
|E'(A, V \setminus A)| - |E'(S, V \setminus  S)|& = \left|( \partial(A) \setminus \partial(S))  \cap E'\right|  - \left|( \partial(S) \setminus \partial(A)) \cap E' \right|\\
& = |E^+ \cap E'| - |E^- \cap E'|. 
\end{align*} 
So we need to bound the probability of the event $|E^+ \cap E'| - |E^- \cap E'| \geq \frac{p \epsilon}{2} {d-1 \choose k-1}\frac{|V|}2$. From the lemma assumption, together with the fact that $|A| = |S| = |V|/2$, we have 
\begin{equation}\label{eq:E+E-_general}
 |E^+| - |E^-| = |E(A, V \setminus A)| - |E(S, V \setminus S)| \geq (1+\epsilon){d-1 \choose k-1}|A|  - {d-1 \choose k-1}|S| = \epsilon {d-1 \choose k-1}\frac{|V|}{2}.
 \end{equation}
By Lemma~\ref{lem:small_cut_general}, 
    $$ \E[|E^+ \cap E'|] = p|E^+| \leq p \cdot C \epsilon{d-1 \choose k-1}\frac{|V|}{2} \qquad \text{and} \qquad \E[|E^- \cap E'|] = p|E^-| \leq p \cdot C\epsilon {d-1 \choose k-1}\frac{|V|}{2} .$$
    Let $\lambda = \frac{p\epsilon}{4}{d-1 \choose k-1} \frac{|V|}{2} .$ Then $\min\{ \lambda, \lambda^2/p|E^-|\}, \min\{ \lambda, \lambda^2/p|E^+|\} \geq \Omega\left(\epsilon p {d-1 \choose k-1}2^{d-1} \right) = \Omega(\epsilon p d^{k-1}2^{d})$, so applying the Chernoff bound (Lemma~\ref{lemma:chernoff}) gives
    $$  \Pr\left[\Bigr| |E^-\cap E'|-p|E^-| \Bigr| \geq \frac{p\epsilon}{4}{d-1 \choose k-1}\frac{|V|}{2} \right] \leq 2e^{-\Omega(\epsilon p d^{k-1}2^{d})}$$
and 
    $$  \Pr\left[\Bigr| |E^+\cap E'|-p|E^+| \Bigr| \geq \frac{p\epsilon}{4}{d-1 \choose k-1}\frac{|V|}{2} \right] \leq 2e^{-\Omega(\epsilon p d^{k-1}2^{d})}. $$

By a union bound, with probability at least $1 - 4e^{-\Omega(\epsilon p d^{k-1}2^{d})}$, it holds that 
\begin{align*}
    |E^+\cap E'| - |E^- \cap E'| & \geq p|E^+| - p|E^-| - \frac{p \epsilon} {2} {d-1 \choose k-1}\frac{|V|}{2}  \\
    & \geq p \epsilon {d-1 \choose k-1 }\frac{|V|}{2}  -  \frac{p \epsilon} {2} {d-1 \choose k-1}\frac{|V|}{2}  \\
    & = \frac{p \epsilon}{2}{d-1 \choose k-1} \frac{|V|}{2} ,
\end{align*}
where second inequality uses Equation \eqref{eq:E+E-_general}. This is the desired bound, which completes the proof. 
\end{proof}

\begin{lem}\label{claim:thm3_claim}
Let $Q=(V,E)$ be a component of $Q_{d,k}$ (as per Definition~\ref{def:component}). Suppose $p = \kappa \log d/d^{k-1}$ for a sufficiently large constant $\kappa$, and let $\epsilon_0  = 2^{-d}/K$, where $K$ is the universal constant from Lemma~\ref{lemma:sparse-close-to-coord_general}. 
Then with probability at least $1 - d^{-100}/2$, the following holds:  For every $\epsilon \geq \epsilon_0$, and  every balanced cut $A\subseteq Q$ of size $|E(A, V \setminus A)| = (1+\epsilon){d-1 \choose k-1}\frac{|V|}{2}$, it holds that
   $$ |E'(A, V \setminus A)| \geq |E'(S, V \setminus S)| + \frac{p \epsilon}{2}{d-1 \choose k-1}|V|/2 ,$$
   where $S$ is the coordinate cut in $Q$ such that $|A \triangle S| \leq K \cdot \epsilon 2^d$ (exists by Lemma~\ref{lemma:sparse-close-to-coord_general}). 
\end{lem}
\begin{proof}
Let $C$ be the hidden constant in the $O$-notation in Lemma~\ref{lemma:cutcounting_general} and let $D$ be the hidden constant in the $\Omega$-notation in Lemma~\ref{lemma:cut_concentration_general}. Suppose that $\kappa$ is a sufficiently large constant.   
Let $\epsilon_i = 2^i \epsilon_0$ for $i = 1, \dots \log(2^d/\epsilon_0)$.
       For every coordinate cut $S =S_{j,b}\cap Q$ with $j \in [d]$ and $b \in \{0,1\}$, and for every $\epsilon_i$,  let $\mathcal{B}(S,\epsilon_i)$ be the event that there exists a balanced cut $A\subseteq Q$ of size $(1+\epsilon_i)|A| \leq |E(A, V \setminus A)| \leq (1 + 2 \epsilon_i){d-1 \choose k-1}|A|$ with $|A \triangle S| \leq K \cdot \epsilon 2^d$ such that 
  $$ |E'(A, V \setminus A)|< |E'(S, V \setminus S)| + \frac{p \epsilon_i}{2}{d-1 \choose k-1}\frac{|V|}{2}.$$
  
  We now show that $\Pr[\mathcal{B}(Q, \epsilon_i)] \leq d^{-102}/8$ for every $i = 1, \dots , \log(2^d/\epsilon_0)$. 
  Let $i \in \{1, \dots , \log(2^d/\epsilon_0)\}$.  For every balanced cut $A$ of size $|E(A, V \setminus A)| \leq (1 + 2 \epsilon_i){d-1 \choose k-1}|A|$ with $|A \triangle S| \leq K \cdot \epsilon_i2^d$, by  Lemma~\ref{lemma:cut_concentration_general},  $$\Pr\left[ |E'(A, V \setminus A)| < |E'(S, V \setminus S)| + \frac{p \epsilon_i}{2}{d-1 \choose k-1}\frac{|V|}{2} \right] \leq 4\exp\left( -D\epsilon_i p d^{k-1}2^{d}\right).$$
By Lemma~\ref{lemma:cutcounting_general}, the  number of such cuts is at most $\exp\left(2^d C\e_i/d \log(d/\e_i)\right)$. Therefore, by a union bound,   
\begin{align*}
\Pr[\mathcal{B}(S, \epsilon_i)] &\leq 4\exp\left(-D\epsilon_i p d^{k-1}2^d\right)\cdot \exp\left( 2^d C\e_i/d \log(d/\e_i)\right)
\\
& = 4\exp\left(2^d\epsilon_i \left(  \frac{C\log(d/\epsilon_i)}{d}-Dpd^{k-1}/2\right)\right)  \\
 & \leq 4 \exp\left( \epsilon_i\left( C \log d \log K -D \kappa \log d /2\right)\right)  && \text{since $\log(1 / \epsilon_i) \leq \log(1 / \epsilon_0) = d \log K$}\\
 &  \leq  4\exp\left(-103\log d\right) && \text{for $\kappa$ sufficiently large, since $\epsilon_i \geq \epsilon_0 = 2^{-d}/K$} \\
&  \leq d^{-102}/8.
\end{align*}
Taking a union bound over the $2d$ possible choices for $S$ and the $\log(2^d/\epsilon_0) \leq 2d$ possible choices for $i$, we get the claim for all $\epsilon \geq \epsilon_0$. 
\end{proof}
 \begin{cor}\label{cor:not_too_small_general}
    Conditioned on the success of the event in Lemma~\ref{claim:thm3_claim}, for every balanced cut $A \subseteq Q$ it holds that 
    $$|E'(A, V \setminus A)| \geq|E'(S, V \setminus S)|,$$
 \end{cor}
    where $S$ is the coordinate with the smallest hamming distance to $A$. 
\begin{proof}
Let $A$ be a balanced cut of size $|E(A,V \setminus A)| = (1+ \epsilon){d-1 \choose k-1}|V|$. 
If $\epsilon \geq  \epsilon_0$, then we are done by Lemma~\ref{claim:thm3_claim}. 
If instead $\epsilon < \epsilon_0$, then  $|E(A, V \setminus A)| < (1+ \epsilon_0){d-1 \choose k-1}|A| $, so by Lemma~\ref{lemma:sparse-close-to-coord_general}, there exists a coordinate cut $S$ such that$$|A \triangle S| < K \epsilon_0 2^{d-1} \leq  1,$$ where the last inequality follows by the setting  $\epsilon_0 = 2^{-d}/K$. But then $A$ is equal to $S$, so we are done.
\end{proof}
\begin{lem}\label{lemma:optvalue_general} 
Let $K$ be the universal constant from Lemma~\ref{lemma:sparse-close-to-coord_general}. 
With probability at least $1-d^{-100}/2$, all coordinate cuts $S$ in $Q$ satisfy
   $$\left| |E'(S, V \setminus S)|  - p{d-1 \choose k-1}2^{d-1}\right| \geq \frac{p}{100K d}{d-1 \choose k-1}\frac{|V|}{2},$$
  and in particular, the optimal value of \eqref{opt_problem_general} is at most $\left( d + \frac{1}{100K}\right){d-1 \choose k-1}\frac{|V|}{2}p$.
\end{lem}
\begin{proof}
    Fix a coordinate cut $S = S_{j,b}\cap Q$. Then $\E[|E'(S, V \setminus S)|] = p{d-1 \choose k-1}\frac{|V|}{2}$, so applying the Chernoff bound (Lemma~\ref{lemma:chernoff}) with $\lambda = \frac{p}{100Kd}{d-1 \choose k-1}\frac{|V|}{2}$ gives

    \begin{align*}
     \Pr\left [\left| |E'(S_{j,b}, V \setminus S_{j,b})|  - p{d-1 \choose k-1}\frac{|V|}{2}\right| \geq \frac{p}{100Kd}{d-1 \choose k-1}\frac{|V|}{2} \right]&  \leq \exp\left(-\frac{1}{3}\frac{1}{100^2K^2d^2}{d-1 \choose k-1}p\frac{|V|}{2}\right) \\
     & = \exp\left(-2^d/\poly d\right) \\
     & \leq d^{-101}/2. 
    \end{align*}
    By a union bound over the $d$ coordinate cuts, the above inequality holds simultaneously for all $S = S_{j,b}\cap Q$ with probability at least $1 - d^{-100}/2.$
    If this holds, then, since $\{S_{j,b}\cap Q : j \in [d], b = 0\}$ is a feasible solution to \eqref{opt_problem_general}, the optimal value of \eqref{opt_problem_general} is at most 
    $$ \sum_{j} |E'(S_{j,0}, V \setminus S_{j,0})| \leq d \cdot p\left( 1 + \frac{1}{100Kd}\right){d-1 \choose k-1}\frac{|V|}{2} = p \left(d +\frac{1}{100K}\right){d-1 \choose k-1}\frac{|V|}{2}.$$
 \end{proof}

\begin{proof}[Proof of Theorem~\ref{thm:general_main}]
The algorithm solves the following optimization problem: 
\begin{equation}\tag{\ref{opt_problem_general}}
 \begin{aligned}
      \min &  \sum_{i = 1}^d |E'(A_i, V \setminus A_i)   \qquad \text{subject to}&& \\
            & |A_i| = \frac{|V|}{2} && \forall i,   \\
     & |A_i \triangle A_j| = \frac{|V|}{2}&&  \forall i \neq j ,\\
    \end{aligned}
\end{equation}
and outputs the optimal solution $A_1,\dots,A_d$. 
\paragraph{Running time.}We can solve  \eqref{opt_problem_general} by enumerating over all feasible families of cuts, of which there are at most $\left({2^d \choose 2^{d-1}}\right)^d = 2^{O\left(2^d d\right)} = 2^{O(n \log n)} $,  and computing the corresponding edge counts, so the running time is $2^{O(n \log n)}$. 

\paragraph{Correctness.} Condition on the success of the events in Lemma~\ref{claim:thm3_claim} and Lemma~\ref{lemma:optvalue_general}. By a union bound, this occurs with probability at least $1- d^{-100}.$

Let $\{A_i\}_{i \in [d]}$ be the optimal solution to~\eqref{opt_problem_general}. For every $i \in [d]$, let $S_i$ denote the coordinate cut in $Q$ with the smallest Hamming distance to $A_i$. We start by proving that this is a matching, i.e., that the set $\{S_i\}_{i \in [d]}$ consists of $d$ different coordinate cuts. Suppose not. Then $S_i = S_j$ or $S_i = \overline S_j$ for some $i \neq j$. If $S_i = S_j = S$, then by triangle inequality, 
\[
|A_i \triangle S| + |A_j \triangle S| \;\geq\; |A_i \triangle A_j| \;=\; \frac{|V|}{2},
\]
so either $|A_i \triangle S_i | \geq |V|/4$ or $|A_j \triangle S_j | \geq |V|/4.$

If instead $S_i = \overline S_j = S$, then again by triangle inequality, 
$$ |A_i \triangle S| + |A_j \triangle \overline S| = |A_i \triangle S| + V - |A_j \triangle S|  \geq 2^d-  |A_i \triangle A_j| = \frac{|V|}{2},$$
so again either $|A_i \triangle S_i| \geq  |V|/4$ or $|A_j \triangle S_j| \geq |V|/4$. 

Let $i$ be the index such that $|A_i \triangle S_i| \geq |V|/4$. Applying Lemma~\ref{lemma:sparse-close-to-coord_general} with $\epsilon = 1/4$ gives $|E(A_i, V \setminus A_i)| \geq (1 + \frac{1}{4K})\frac{|V|}{2}$. 
From Lemma~\ref{claim:thm3_claim} and Lemma~\ref{lemma:optvalue_general}, we get
\begin{align*}
|E'(A_i, V \setminus A_i)| & \geq |E'(S_i, V \setminus S_i)| + \frac{p}{8K}{d-1 \choose k-1}\frac{|V|}{2}   && \text{by Lemma~\ref{claim:thm3_claim}} \\
& \geq p  \left(1 -  \frac{1}{100K}\right){d-1 \choose k-1}\frac{|V|}{2}  + \frac{p}{8K}{d-1 \choose k-1}\frac{|V|}{2}   && \text{by Lemma~\ref{lemma:optvalue_general}} \\
& \geq p\left(1 +  \frac{1}{10K}\right){d-1 \choose k-1}\frac{|V|}{2} . 
\end{align*}
Furthermore, from Lemma~\ref{claim:thm3_claim} and Lemma~\ref{lemma:optvalue_general}, for every $j \neq i$, 
\begin{align*}
    |E'(A_j, V \setminus A_j)| & \geq |E'(S_j, V \setminus S_j)| \geq p \left(1 - \frac{1}{100K  d}\right){d-1 \choose k-1}\frac{|V|}{2}.    
\end{align*}
But then summing over all $j \in [d]$ gives
\begin{align}
    \sum_{j \in [d]}|E'(A_j, V \setminus A_j)| & \geq p \left(1 +  \frac{1}{10K}\right){d-1 \choose k-1}\frac{|V|}{2}  + p (d-1)\left(1 -  \frac{1}{100Kd}\right){d-1 \choose k-1}\frac{|V|}{2}  \\
    &>p \left(d+ \frac{1}{100K}\right){d-1 \choose k-1}\frac{|V|}{2} ,
\end{align}
which is a contradiction, since objective value of \eqref{opt_problem_general} is at most $p\left(d+ \frac{1}{100K}\right)|V|/2$, by Lemma~\ref{lemma:optvalue_general}. Thus, the set $\{S_i\}_{i \in [d]}$ must contain $d$ distinct coordinate cuts. 

So now suppose that we have a matching, i.e. that the set $\{S_i\}_{ i \in [d]}$ contains $d$ distinct coordinate cuts. Then $\{S_i\}_{i \in [d]}$ is a feasible solution to \eqref{opt_problem_general}. 
Suppose for contradiction that $|A_i \triangle S_i| \geq 1$ for some $i$. Then by 
\Cref{lemma:sparse-close-to-coord_general}, $|E(A_i, V \setminus A_i)| \geq \left(1+1/K \right){d-1 \choose k-1}|A|$, so by Lemma~\ref{claim:thm3_claim}, 
$$| E'(A_i, V \setminus A_i)| > |E'(S_i, V \setminus S_i)|. $$ 
For every $j \neq i$, by Lemma~\ref{cor:not_too_small_general} applied to $A_j$, 
$$| E'(A_j, V \setminus A_j)| \geq |E'(S_j, V \setminus S_j)|.$$
But this gives
   $ \sum_{j \in  [d]}|E'(A_j, V \setminus A_j)| 
    > \sum_{j =1}^d|E'(S_j, V \setminus S_j)|$, 
 contradicting the optimality of $\{A_i\}_{i \in [d]}$.
\end{proof}

\subsection{Proof of Lemma~\ref{lemma:fkn_even}}\label{sec:fkn_pf}
In this section, we prove Lemma~\ref{lemma:fkn_even}, which is restated below for the convenience of the reader. 

One might hope to deduce Lemma~\ref{lemma:fkn_even} from the standard FKN theorem (\Cref{thm:fkn}) by identifying the even subcube $Q_d^E$ with $\{0,1\}^{d-1}$. For example, a natural approach is to define a function $g : \{0,1\}^{d-1} \rightarrow \{0,1\}$ by setting $g(x)=f(x^E)$ where $x^E=(x,0)$ if $|x|$ is even and $x^E=(x,1)$ otherwise. Then $g$ is boolean and $\widehat g(S) = \widehat f(S) + \widehat f([d] \setminus S) = 2 \widehat f(S)$ for every $S \subseteq [d-1]$, so the higher level Fourier mass of $g$ satisfies $\sum_{S \subseteq [d-1] : |S| \geq 2} \widehat g(S)^2 = 4 \sum_{S \subseteq [d-1] : |S| \geq 2} \widehat f(S)^2$. However, the hypothesis in Lemma~\ref{lemma:fkn_even} only controls $\sum_{2\le|S|\le d-2}\widehat f(S)^2$ and does not give information about the $|S|=d-1$ term $\widehat f([d-1])^2$. Thus, such a reduction does not immediately allow for an application of \Cref{thm:fkn}.  Instead, we prove the lemma directly, adapting Proof II in \cite{fkn} to our setting. 
\evenfkn*

\begin{proof}Let
    $$ S_1 \coloneqq \sum_{|T| \leq 1} \widehat{f}(T)\chi_T  \qquad \text{ and }\qquad  S_2 \coloneqq \sum_{|T| \geq d-1} \widehat{f}(T)\chi_T$$
    be the projection of $f$ onto the bottom two levels and top two levels, respectively, and let 
    $$L  \coloneqq \sum_{2 \leq |T| \leq d-2} \widehat{f}(T)\chi_T$$ 
    be the projection onto the remaining levels. Then 
    $f  = S_1 + S_2 + L.$
    Let   $$\epsilon:= \langle L, L \rangle \leq \delta$$ be the Fourier mass on the middle levels, and let
    $$R_1 \coloneqq  2 S_1^2 - S_1 \qquad \text{ and } \qquad  R_2 \coloneqq  2 S_2^2 - S_2.$$ 
We now compute the Fourier coefficients of $R_1$ and $R_2$. Since $f$ is supported only on the even component, it follows that the Fourier coefficients of $f$ are symmetric, in the sense that 
\begin{equation}\label{eq:sym}
\widehat f(T) = \widehat f([d] \setminus T) \qquad \forall T \subseteq [d]. 
\end{equation}
Indeed, for every $T \subseteq [d]$, 
\begin{align*}
\widehat f(T) & = 2^{-d}\sum_x f(x)(-1)^{\langle x, \1_T\rangle} \\
& = 2^{-d}\sum_{x: |x| \equiv 0 \pmod{2} } f(x)(-1)^{\langle x, \1_T\rangle}  \\
& = 2^{-d}\sum_{x: |x| \equiv 0 \pmod{2} } f(x)(-1)^{\langle x, \1_T \oplus \1_{[d]}\rangle}  \\
& = \widehat f([d]\setminus T). 
\end{align*}
Therefore, 
$$ \langle S_1,S_1 \rangle  =  \langle S_2, S_2 \rangle.$$ 
Furthermore, by orthogonality of $S_1, S_2$ and $L$, we have 
\begin{equation*}
    \frac{1}{4} = \langle f ,f \rangle =  \langle S_1,S_1 \rangle + \langle S_2, S_2 \rangle + \langle L, L \rangle .
\end{equation*}
This yields $\langle S_1 , S_1\rangle + \langle S_2, S_2\rangle = \frac{1}{4} - \epsilon$, and hence
	\begin{equation} \label{eq:S1}
	\langle S_1^2 , \chi_\emptyset  \rangle = \langle S_1 , S_1 \rangle=  \langle S_2 , S_2 \rangle = \langle S_2^2, \chi_{\emptyset}\rangle =  \frac{1}{8}-\frac \epsilon 2.
	\end{equation}
   Since $f$ is boolean, we also have 
\begin{equation}\label{eq:f_emptyset}
     \langle S_1 , \chi_{\emptyset} \rangle =\langle f,  \chi_{\emptyset}\rangle = \langle f,f\rangle =   \frac{1}{4}.
\end{equation}
    Combining \Cref{eq:S1} and \Cref{eq:f_emptyset} gives 
    
	$$\widehat{R_1}(\emptyset) = \langle R_1 ,\chi_{\emptyset}\rangle = \langle2S_1^2 - S_1, \chi_{\emptyset} \rangle = 2\left ( \frac{1}{8} - \frac{\epsilon}{2} \right) - \frac{1}{4} = -\epsilon.$$
For every  $i$, we have 
$$\widehat{R_1}(i) = \langle R_1, \chi_i\rangle = 2 \langle S_1^2, \chi_i\rangle - \langle S_1, \chi_i\rangle = 4 \widehat{f}(i)\widehat f(\emptyset) - \widehat f(i) =0,$$ and for every $i \neq j$, we have   $$\widehat{R_1}(ij)=\langle R_1, \chi_{ij}\rangle  =  2\langle  S_1^2 , \chi_{ij} \rangle -\langle  S_1, \chi_{ij} \rangle = 4\widehat{f}(i)\widehat{f}(j).$$ Finally, for $|T| \geq 3$, we have $\widehat R_1(T)  = \langle R_1, \chi_T\rangle  =  2\langle  S_1^2 , \chi_T \rangle -\langle  S_1, \chi_T \rangle = 0$. 
This yields
	$$R_1 = - \epsilon \chi_{\emptyset} + 4 \sum_{i < j} \widehat{f}(i)\widehat{f}(j)\chi_{ij}.$$
Similarly, using Equation \eqref{eq:sym}, we have
		$$R_2= -  \epsilon  \chi_{[d]} + 4 \sum_{i < j} \widehat{f}(i)\widehat{f}(j)\chi_{[d] \setminus \{i,j\}}. $$

\begin{claim}\label{claim:small_R}
		$	\langle R_1 + R_2, R_1 + R_2 \rangle \leq O(\epsilon).$
	\end{claim}
\begin{cor}\label{cor:f_hat} There exists $i \in [d]$ such that 
$$\widehat{f}(i) =\widehat{f}([d] \setminus \{i\})  = \pm \left(\frac{1}{4}-O(\epsilon)\right)$$
\end{cor}

We begin by showing how the lemma follows from Corollary~\ref{cor:f_hat}, we then derive Corollary~\ref{cor:f_hat} from Claim~\ref{claim:small_R}, and finally prove Claim~\ref{claim:small_R} itself.    

One can verify that
$$\1_E \cdot x_i = \frac{1}{4}\left(\chi_{\emptyset} + \chi_{[d]} -  \chi_{i} - \chi_{[d] \setminus \{i\}} \right) \ \text { and } \ \1_E \cdot (1-x_i) = \frac{1}{4}\left(\chi_{\emptyset} + \chi_{[d]} +  \chi_{i} + \chi_{[d] \setminus \{i\}} \right).$$

Suppose that $\widehat f(i) =  \widehat f([d] \setminus \{i\}) = \frac{1}{4} - O(\epsilon)$. 
By Equation \eqref{eq:f_emptyset} and Equation \eqref{eq:sym},  $\widehat{f}(\emptyset) = \widehat{f}([d]) = \frac{1}{4}$.  By  the lemma assumption, $\sum_{2 \leq |T| \leq d-2} \widehat{f}(T)^2 = \delta$. Combining this with Corollary~\ref{cor:f_hat} gives
\begin{align*}
   \left \|f - \1_E \cdot (1-x_i) \right\|^2 & = \left\|f -  \frac{1}{4}\left(\chi_{\emptyset} + \chi_{[d]} +  \chi_{i} + \chi_{[d] \setminus \{i\}} \right) \right\|^2 \\
   & =\left(\widehat f(\emptyset)-\frac{1}{4}\right)^2 + \left(\widehat f([d])-\frac{1}{4}\right)^2  + \left(\widehat f(i)-\frac{1}{4}\right)^2 + \left(\widehat f([d] \setminus \{i\})-\frac{1}{4}\right)^2 + \sum_{2 \leq |T| \leq d-2} \widehat{f}(T)^2 \\
   & \leq O(\epsilon^2) + \sum_{2 \leq |T| \leq d-2} \widehat{f}(T)^2   \qquad \text{by Corollary~\ref{cor:f_hat}}  \\
   & \leq O(\delta) \qquad \qquad\qquad\qquad\qquad\text{by the Lemma assumption and using $\epsilon \leq \delta$}.
\end{align*}

Suppose instead that $\widehat f(i) =  \widehat f([d] \setminus \{i\}) = -\left(\frac{1}{4} - O(\epsilon)\right)$. Then, by a similar argument, 
$$  \left\|f - \1_E \cdot x_i) \right\|^2  = \left\|f -  \frac{1}{4}\left(\chi_{\emptyset} + \chi_{[d]} -  \chi_{i} - \chi_{[d] \setminus \{i\}} \right) \right\|^2 \leq O(\delta), $$
as required. We now show how Corollary~\ref{cor:f_hat} follows from Claim~\ref{claim:small_R}.

\begin{proof}[Proof of Corollary~\ref{cor:f_hat}:]
   
By Equation \eqref{eq:sym},  we have $\widehat{f}(i) =\widehat{f}([d] \setminus \{i\}) $ for all $i$. We now prove that there exists $i \in [d]$ such that $\widehat{f}(i) = \pm \left( \frac{1}{4} - O(\epsilon)\right)$. 
Note that $|\widehat f(i)| \leq 1/4$ for all $i$, since 
$$2  \widehat f(i)^2  = \widehat f(i)^2 + \widehat f([d] \setminus \{i\})^2\leq \|f \|^2_2 - \widehat f(\emptyset )^2 - \widehat f ([d])^2 = \|f\|^2_2 - 2 \widehat f(\emptyset) = 1/4 - 2/16 = 1/8, $$
where the first and third transition used \Cref{eq:sym}. Therefore, it suffices to show that $|\widehat f(i)| \geq 1/4 - O(\epsilon)$ for some $i$, or equivalently that $\widehat f(i)^2 \geq 1/16 - O(\epsilon)$.

Since $R_1$ and $R_2$ are orthogonal to each other,  $$\langle R_1 + R_2, R_1 + R_2 \rangle = \langle R_1,  R_1 \rangle + \langle R_2,  R_2 \rangle = \sum_T \left(\widehat {R_1}(T)^2 + \widehat{R_2}(T)^2\right).$$ 
Therefore 
\begin{align*}
    \sum_{i < j} \widehat{f}(i)^2 \widehat{f}(j)^2 \leq \left( 4 \sum_{i <j} \widehat{f}(i)\widehat{f}(j) \right)^2 = \widehat R_1(ij)^2 \leq \sum_T \left(\widehat {R_1}(T)^2 + \widehat{R_2}(T)^2\right) = \langle R_1 + R_2, R_1 + R_2 \rangle \leq O(\epsilon),
\end{align*}
where the last inequality follows by Claim~\ref{claim:small_R}. On the other hand, by Equation \eqref{eq:S1}, we have  $\langle S_1, S_1\rangle = \widehat{f}(\emptyset)^2 + \sum_i \widehat{f}(i)^2 = 1/8-\epsilon/2,$ so 
$$\sum_i \widehat{f}(i)^2  = \frac{1}{8} - \frac{\epsilon}{2}- \widehat f(\emptyset)^2 = \frac{1}{16}-\frac{\epsilon}{2}.$$ 
Combining the above two equations gives

$$ \left(\frac{1}{16}-\frac{\epsilon}{2} \right)^2  = \left( \sum_i \widehat f(i)^2\right)^2 \leq 2\sum_{i<j}\widehat f(i)^2 \widehat f(j)^2 + \max_i \widehat f(i)^2 \cdot \sum_{j}\widehat f(j)^2 \leq O(\epsilon)  + \left(\frac{1}{16}-\frac{\epsilon}{2}\right) \max_i \widehat{f}(i)^2,$$
Rearranging gives $\max_i \widehat f(i)^2 \geq 1/16-O(\epsilon)$, as required. 
\end{proof}
Finally, we prove Claim~\ref{claim:small_R}. 
	\begin{proof}[Proof of Claim~\ref{claim:small_R}] We start by bounding the probability that $(R_1(x) + R_2(x))^2$ is large. 
\begin{claim}\label{claim:p_alpha} Given $\alpha \in (0,1]$, let $p_{\alpha} = \Pr[(R_1 + R_2)^2 > \alpha^2]. $ Then 
				$$p_{\alpha}  \leq \frac{16 \e}{\alpha^2}.$$
\end{claim}
\begin{proof}Recall that 
\begin{equation}\label{eq:Rsum}
    R_1 + R_2 = 2S_1^2 - S_1 + 2S_2^2 - S_2 = -  \epsilon \left( \chi_{\emptyset } + \chi_{[d]}\right) + 4 \sum_{i \neq j }\widehat f(i)^2\widehat f(j)^2\left(\chi_ i + \chi_{[d]\setminus \{i\}} \right).
\end{equation}
\noindent 
Note that for every  $x \in \{0,1\}^d$ and every $T \subseteq [d]$, 
\begin{equation}\label{eq:antisym}
\chi_{T}(x) = (-1)^{\langle x, \1_T\rangle } =(-1)^{|x|} (-1)^{\langle x, \1_T\rangle \oplus \1_{[d]}}=  (-1)^{|x|}\chi_{[d] \setminus T}(x). 
\end{equation}
From \Cref{eq:Rsum} and \Cref{eq:antisym}, we see that $R_1(x) + R_2(x) =0 $ whenever $|x|$ is odd, and in particular 
\begin{equation}\label{eq:Roddbnd}
    \Pr \left[ (R_1(x) + R_2(x))^2 >\alpha^2 \mid |x| \text{ is odd}\right] = 0
\end{equation}

We now consider $x$ with even Hamming weight. For such $x$, by Equation \eqref{eq:sym} and \Cref{eq:antisym}, we have $S_1(x) = S_2(x)$, so $2S_1(x) = S_1(x)+S_2(x) = f(x)-L(x)$. Thus,
$$R_1(x)+R_2(x) = 4S_1(x)^2-2S_1(x) = ( f(x)-L(x))^2 - ( f(x)-L(x)) = L(x)^2+L(x)(1-2f(x)), $$ 
where the last equality uses that $f^2 = f$ by the assumption that $f$ is boolean. By analyzing the cases $f(x) =1$ and $f(x)=0$, we see that  $|R_1(x) + R_2(x)| < \alpha $ whenever $|L(x)| \leq \alpha/4$. Therefore, 

\begin{equation}\label{eq:Revenbnd}
     \Pr \left[ (R_1(x) + R_2(x))^2 >\alpha^2 \mid |x| \text{ is even}\right]  \leq \Pr \left[ L(x)^2 >\alpha^2/16 \mid |x| \text{ is even}\right].
\end{equation}
To continue, note that $\mathbb{E}[L^2(x)] = \langle L ,L \rangle = \epsilon$, and that  by \Cref{eq:sym} and \Cref{eq:antisym},  $L(x) = 0$ whenever $|x|$ is odd. Therefore, by Markov's inequality, 
\begin{equation}\label{eq:Lbnd}
    \Pr \left[ L(x)^2 >\alpha^2/16 \mid |x| \text{ is even}\right]  = 2   \Pr \left[ L(x)^2 >\alpha^2/16 \right] \leq \frac{32\epsilon}{\alpha^2}. 
\end{equation}
Combining \Cref{eq:Roddbnd}, \Cref{eq:Revenbnd} and \Cref{eq:Lbnd} gives the claim. 
\end{proof}

				Next, we will show that $\langle R_1 + R_2, R_1 + R_2 \rangle =  \mathbb{E}[(R_1+R_2)^2] = O(\epsilon).$ To this end, we will need the Bonami-Beckner hypercontractive  inequality: 
\begin{thm}[Bonami-Beckner hypercontractive  inequality \cite{Beckner}\cite{Bonami}] \label{thm:contractivity}
	Let $f:\{0,1\}^d \rightarrow \mathbb{R}$ be a function which is a linear combination of $\{\chi_T : |T| \leq \ell \}$. Let $p > 2$. Then 
					$$\|f\|_p \leq( \sqrt{p-1})^\ell \|f\|_2.$$
	\end{thm}
\begin{claim}
$$ \mathbb{E}[(R_1+R_2)^2] \leq \frac{\alpha^2}{1-72\sqrt{\e}/\alpha}$$
\end{claim}
\begin{proof}
Applying Theorem~\ref{thm:contractivity} to $R_1$ with $p=4$ and $\ell =2$, yields $\mathbb{E}[R_1^4] \leq 81 \mathbb{E}[R_1^2]^2 $. By \Cref{eq:sym} and \Cref{eq:antisym}, we have $R_1(x)^2 = R_2(x)^2$ for all $x$. Writing $R:= R_1 + R_2$, this gives
\begin{equation}\label{eq:R}
		\mathbb{E}[R^4]  =  \mathbb{E}[(R_1+R_2)^4] = \mathbb{E}[R_1^4 + 4R_1^3R_2 + 6R_1^2R_2^2 + 4R_1R_2^3 + R_2^4]  \leq 16 \mathbb{E}[R_1^4] \leq 16 \cdot 81 \mathbb{E}[R_1^2]^2 = 4 \cdot 81  \mathbb{E}[R^2]^2 , 
\end{equation}
so 
		\begin{align*}
			\mathbb{E}[R^2] &= (1-p_{\alpha})\mathbb{E}[R^2 | R^2 \leq \alpha^2] +  p_{\alpha}\mathbb{E}[R^2 | R^2 > \alpha^2] &&  \\
		& 	\leq \alpha^2 + p_{\alpha} \sqrt{\mathbb{E}[R^4 | R^2 > \alpha^2]} &&  \\
		& \leq \alpha^2  + p_{\alpha} \sqrt{ \frac{\mathbb{E}[R^4]}{p_{\alpha}}} && \\
		& \leq \alpha^2  + \sqrt{p_{\alpha}}18\E[R^2]  && \text{by Equation \eqref{eq:R}}\\
		& \leq \alpha^2  + \frac{4 \sqrt{\e}}{\alpha}18\E[R^2] && \text{by Claim~\ref{claim:p_alpha}.}  
		\end{align*} Rearranging gives $\E[R^2] \leq \frac{\alpha^2}{1-72\sqrt{\e}/\alpha}.$
		\end{proof}
		Thus,  $\langle R_1 + R_2, R_1 + R_2 \rangle = \mathbb{E}[(R_1+R_2)^2] \leq \frac{\alpha^2}{1-72\sqrt{\e}/\alpha}$ for every $\alpha \in (0,1]$.  Setting  $\alpha = 144 \sqrt{\epsilon}$, gives  $\langle R_1 + R_2, R_1 + R_2 \rangle \leq O(\e), $ as required. 
		\end{proof}
\end{proof}

\bibliographystyle{alpha} 
\bibliography{bibliography.bib}
\end{document}

%% file: preamble.tex
\theoremstyle{plain}
\newtheorem{thm}{\protect\theoremname}
\theoremstyle{plain}
\newtheorem{claim}[thm]{\protect\claimname}
\theoremstyle{plain}

\theoremstyle{plain}
\newtheorem{lem}[thm]{\protect\lemmaname}
\theoremstyle{plain}
\newtheorem{cor}[thm]{\protect\corollaryname}
\theoremstyle{definition}
\newtheorem{defn}[thm]{\protect\definitionname}
\theoremstyle{definition}

\theoremstyle{definition}
\newtheorem{rem}{\protect\remarkname}
\theoremstyle{plain}

\makeatother

\providecommand{\claimname}{Claim}
\providecommand{\lemmaname}{Lemma}
\providecommand{\propositionname}{Proposition}
\providecommand{\theoremname}{Theorem}
\providecommand{\corollaryname}{Corollary}
\providecommand{\definitionname}{Definition}
\providecommand{\assumptionname}{Assumption}
\providecommand{\remarkname}{Remark}

\global\long\def\R{\mathbb{R}}

\global\long\def\1{\mathbbm{1}}
\global\long\def\poly#1{\mathrm{poly}\left(#1\right)}

\global\long\def\e{{\epsilon}}